\newif\ifNOTES
\NOTESfalse

    \documentclass[11pt]{article}
    \usepackage[margin=1in]{geometry}
    \usepackage[hidelinks]{hyperref}

    \usepackage{enumitem}
    
    \usepackage[utf8]{inputenc} %
    \usepackage[T1]{fontenc}    %
    \usepackage{url}            %
    \usepackage{booktabs}       %
    \usepackage{amsfonts}       %
    \usepackage{nicefrac}       %
    \usepackage{microtype}      %
    \usepackage{xcolor}         %

    \usepackage[ruled, noline, noend]{algorithm2e}

    \usepackage[at]{easylist}
    \usepackage{amsmath,amssymb,amsfonts,amsthm}
    \usepackage{thm-restate}

    \usepackage[colorinlistoftodos,textsize=tiny]{todonotes}
    \setuptodonotes{color=orange!30}
    \def\withcolors{1}

    \ifNOTES
        \def\withnotes{1}
    \else
        \def\withnotes{0}
    \fi

\def\eps{\ve}
\renewcommand{\epsilon}{\ve}
\def\ve{\varepsilon}

\newcommand{\sub}{\subseteq}

\newcommand{\ideas}{\mathsf{ideas}}
\newcommand{\id}{\mathsf{id}}
\newcommand{\Ideas}{\cI}
\newcommand{\kv}{\mathsf{kv}}
\renewcommand{\u}{\mathsf{u}}

\newcommand{\scrub}{\mathsf{scrub}}

\newcommand{\subsim}{\mathsf{SubSim}}

\newcommand{\train}{\mathsf{Train}}
\newcommand{\coin}{\mathsf{coin}}

\newcommand{\W}{\cW}
\newcommand{\Graph}{G}
\newcommand{\Edges}{E}

\newcommand{\C}{\cC}
\newcommand{\aux}{\mathsf{aux}}

\newcommand{\Models}{\cP}

\newcommand{\DD}{\mathfrak{D}}
\newcommand{\D}{D}
\newcommand{\B}{\mathfrak{B}}

\newcommand{\pdot}[1]{p(\cdot|#1)}

\newcommand{\pdotx}{\pdot{x}}
\newcommand{\dmax}{\Delta_{\mathrm {max}}}
\newcommand{\safe}{\mathsf{safe}}
\newcommand{\shardsafe}{\mathsf{sharded\mbox{-}safe}}

\newcommand{\CP}{\mathsf{CP}}
\newcommand{\CPk}{\mathsf{CP\text{-}k}}

\SetKwInput{KwShard}{Shard}
\SetKwInput{KwLearn}{Learn}
\SetKwInput{KwInput}{Input}
\SetKwInput{KwOutput}{Output}
\SetKwInput{KwParams}{Parameters}

\newcommand{\pr}[2][]{\mathrm{Pr}\ifthenelse{\not\equal{}{#1}}{_{#1}}{}\!\left[#2\right]}

\newcommand{\dtv}{d_{\mathrm {TV}}}

\DeclareMathOperator{\argmax}{argmax}
\providecommand{\poly}{\operatorname*{poly}}

\newcommand{\Bern}{\mathsf{Bernoulli}}

\theoremstyle{plain}
\newtheorem{theorem}{Theorem}
\newtheorem*{theorem*}{Theorem}
\newtheorem{desideratum}{Desideratum}

\newtheorem{proposition}[theorem]{Proposition}

\newtheorem{lemma}[theorem]{Lemma}
\newtheorem{claim}[theorem]{Claim}
\newtheorem{corollary}[theorem]{Corollary}

\newtheorem{definition}[theorem]{Definition}

\newtheorem{example}[theorem]{Example}
\numberwithin{example}{section}
\numberwithin{theorem}{section} 
\numberwithin{nontheorem}{section} 
\numberwithin{proposition}{section} 
\numberwithin{observation}{section} 
\numberwithin{fact}{section} 
\numberwithin{lemma}{section} 
\numberwithin{claim}{section} 
\numberwithin{corollary}{section} 
\numberwithin{case}{section} 
\numberwithin{dfn}{section} 
\numberwithin{definition}{section} 
\numberwithin{openquestion}{section} 
\numberwithin{res}{section} 

\theoremstyle{definition} %
    
    \newtheorem{remark}[theorem]{Remark}
    \numberwithin{appnote}{section}
    \numberwithin{remark}{section} 
\theoremstyle{plain}

\ifnum\withcolors=1

\else

\fi

\ifnum\withnotes=1

\else

\fi
\newcommand{\ignore}[1]{\leavevmode\unskip} %

\newcommand{\cC}{\mathcal{C}}

\newcommand{\cI}{\mathcal{I}}

\newcommand{\cP}{\mathcal{P}}

\newcommand{\cW}{\mathcal{W}}
\newcommand{\cX}{\mathcal{X}}
\newcommand{\cY}{\mathcal{Y}}

\renewcommand{\tt}[1]{\texttt{#1}}

\newcommand{\Supp}{\mathrm{Supp}}

\title{Blameless Users in a Clean Room:\\Defining Copyright Protection for Generative Models}
\author{Aloni Cohen\\Department of Computer Science \\University of Chicago\\\texttt{aloni@uchicago.edu}}
\begin{document}
\maketitle

\begin{abstract}
    Are there any conditions under which a generative model's outputs are guaranteed not to infringe  the copyrights of its training data? 
    This is the question of ``provable copyright protection'' first posed in \cite{VyasKB23}. They define \emph{near access-freeness} (NAF) and propose it as sufficient for protection.
    This paper revisits the question and establishes new foundations for provable copyright protection---foundations that are firmer both technically and legally.
    First, we show that NAF alone does not prevent infringement. In fact, NAF models can enable verbatim copying, a blatant failure of copyright protection that we dub being \emph{tainted}.
    Then, we introduce our \emph{blameless copyright protection} framework for defining meaningful guarantees, and instantiate it with \emph{clean-room copyright protection}. Clean-room copyright protection allows a user to control their risk of copying by behaving in a way that is unlikely to copy in a counterfactual ``clean-room setting.''
    Finally, we formalize a common intuition about differential privacy and copyright by proving that DP implies clean-room copyright protection when the dataset is \emph{golden}, a copyright deduplication requirement.
\end{abstract}

\section{Introduction}
A user of a generative model is worried about unwitting copyright infringement. She is worried that the model's outputs might resemble copyrighted works in the training data through no fault of her own, exposing her to  legal liability.
The user would like some assurance that this won't happen.
We ask: \textbf{What assurances can the model provider give, and under what conditions?}

Ideally, the generative model would {never} reproduce copyrighted work. That's unrealistic. 
Typical models can be prompted to generate copyrighted output:
$\tt{print the following \underline{\phantom{foob}}}$ \cite{VyasKB23,lee2024talkin}.
It's also unnecessary. In the example, the copying is clearly the user's fault. 

Instead, we should protect blameless users from inadvertent infringement. {If the model reproduces copyrighted work, it should be because the user induced it to} (or was very unlucky).
This guarantee should satisfy a careful, honest user. 
As long as the user's use of the model does not itself cause the infringement, consciously or subconsciously,
then the result is unlikely to infringe.

Our goal is to formalize such a guarantee---in a mathematically and legally rigorous way---and to study its feasibility. We're after ``provable copyright protection'' at deployment time, first studied by Vyas, Kakade, and Barak \cite{VyasKB23}. They propose \emph{near access-freeness} (NAF) as an answer.

\paragraph{Our contributions} This paper establishes new foundations for provable copyright protection---foundations that are firmer both mathematically and legally than prior work.
    \begin{enumerate}
\item \textbf{We prove that \emph{near access-free} (NAF) models can enable verbatim copying.} 
NAF is the first (and only other) attempt at a mathematical definition intended to offer ``provable copyright protection'' for generative models~\cite{VyasKB23}. The proof leverages NAF's lack of protection against multiple prompts or data-dependent prompts.

\item \textbf{We define \emph{tainted models}, which enable users to reproduce verbatim training data} despite knowing nothing about the underlying dataset, a blatant failure of copyright protection~(Section~\ref{sec:tainted}). A meaningful definition of provable copyright protection must exclude tainted models. NAF does not. 

\item \textbf{We introduce a framework for defining copyright protection guarantees, called \emph{blameless copyright protection}}~(Section~\ref{sec:blameless}). 
It protects \emph{blameless} users---those who don't themselves induce infringement---from unwitting copying.

\item \textbf{We define \emph{clean-room copyright protection (($\kappa,\beta)$-clean)}, a first instantiation of our framework}~(Section~\ref{sec:clean-room}), drawing inspiration from \emph{clean-room design}. A training algorithm is $(\kappa,\beta)$-clean if, for every user who copies in a (counterfactual) ``clean-room environment'' with probability $\le \beta$, the probability of copying in the real world is $\le\kappa$. 
Clean-room copyright protection lets users choose their tolerance for risk $\kappa$ and then tune $\beta$ accordingly.
It also excludes tainted models under mild assumptions.

\item \textbf{We prove that \emph{differentially private} (DP) training implies clean-room copyright protection}, formalizing a common intuition about DP and copyright (Section~\ref{sec:DP}). This holds when the dataset is \emph{golden}, a copyright deduplication requirement. 
Thus, DP provides a way to bring copyrighted expression ``into the clean room'' without tainting the model. 
\end{enumerate}

This paper does not offer a legal analysis of where liability would lie under existing law. 
Copyright infringement is generally treated as a strict liability tort, suggesting that the user may always liable.
This paper is premised on the possibility that strict liability may sometimes be unjust when generative AI is involved. 
The simple fact is that there is no feasible, general way for a user to determine with confidence whether a particular model output infringes or not. Some legal scholars agree~\cite{chatterjee2019minds,goodyear2025artificial}, reasoning from caselaw requiring ``some element of volition or causation.''\footnote{\emph{Religious Tech.\ Center v.\ Netcom On-Line Comm.,} 907 F.\ Supp.\ 1361 (N.D.\ Cal.\ 1995)}

Instead of advancing such legal arguments,
 our paper develops a mathematical framework for preventing this sort of unwitting infringement---where holding the user liable would be unjust---and proposes an indemnification policy for assigning fault when copying does inevitably occur (Section~\ref{sec:discussion}).

\paragraph{A reader's guide}
    This paper can be read linearly from beginning to end. Depending on your interests, you might consider reading out of order at first, referencing earlier sections as needed. 
    Sections~\ref{sec:intro:notation},~\ref{sec:legal-interface}, and~\ref{sec:interaction} are required for anything that follows them. 
    For near access-freeness and its limitations: read Sections~\ref{sec:naf}--\ref{sec:tainted}.
    For our approach to defining meaningful copyright protection: read Sections~\ref{sec:tainted}--\ref{sec:clean-room}.
    For the protection offered by differential privacy: start with Section~\ref{sec:DP} for the high level claims, then read Section~\ref{sec:clean-room} to understand the precise formalization.
    For discussion of practical and policy considerations: read Section~\ref{sec:discussion}.

\paragraph{Related work}
There is a lot of recent work on copyright questions for generative AI \cite{cooper2023report,lee2024talkin,henderson2023foundation}, and the first generation of cases are working their way through the courts (e.g., \emph{New York Times v Microsoft}).
But most are only tangentially related to goal of stating and proving formal guarantees against copyright infringement.

Vyas, Kakade, and Barak were the first to propose a mathematical property---near access-freeness (NAF)---aimed at ``preventing deployment-time copyright infringement'' \cite{VyasKB23}.
Elkin-Koren, Hacohen, Livni, and Moran argue that copyright cannot be ``reduced to privacy'' \cite{elkinkoren2024copyright}. Referring to both NAF and DP by umbrella term ``algorithmic stability'', they argue that the sort of provable guarantees that \cite{VyasKB23} and this paper seek do not capture copyright's complexities. 
Li, Shen, and Kawaguchi propose and empirically evaluate an attack called VA3 against NAF~\cite{li2024va3}. They provide good evidence that the $\CPk$ algorithm of \cite{VyasKB23} may not prevent infringement, but some uncertainty remains (see Appendix~\ref{app:sec:va3}).
We discuss these three papers at length. There is also work on new NAF algorithms \cite{golatkar2024cpr,chen2024randomization}.

Scheffler, Tromer, and Varia \cite{scheffler2022formalizing} give a complexity-theoretic account of substantial similarity, and a procedure for adjudicating disputes. In contrast, we treat  substantial similarity as a black box.
That differential privacy might protect against infringement has been suggested in~\cite{bousquet2020synthetic,henderson2023foundation,VyasKB23,elkinkoren2024copyright,chen2024randomization}. 

Livni, Moran, Nissim, and Pabbaraju \cite{livni2024credit} study the problem of attributing credit to inputs of an algorithm that influence its output, motivated by copyright. They define a condition under which a given input need not be credited (with credit required otherwise). Under their definition, a differentially private algorithm never has to credit its inputs---a manifestation of the common intuition that differential privacy prevents copying which we discuss in Section~\ref{sec:DP}.

\subsection{On the use of legal terminology}
    This paper is about copying-as-in-copyright, not copying broadly construed. It is motivated by legal and policy questions coming from copyright law and draws on US copyright law specifically. For this reason, we frame our inquiry as about ``copyright protection.'' 
    On the other hand, we treat the legal concepts as exogenous and
    until Section~\ref{sec:clean-room} our paper is agnostic as to whether training data is in-copyright or not.
    So in some ways our paper can be read as about copying broadly.
    
    Where we use legal terms of art---substantial similarity, the ideas-expression dichotomy, and derivative works---we intend to invoke the legal meaning in the US. Where we wish to avoid doing so, we try to use different words. For example, using variants of ``blame'' instead of fault, liability, innocence, guilt, etc. Any use of legal terms without invoking their legal meaning is unintended.

    We make no claims or conclusions of law. In particular, we do not intend to suggest that our definitions or techniques provide an absolute shield against liability for infringing outputs, nor are we trying to reduce copyright to differential privacy (Appendix~\ref{sec:niva}). 
    This might seem at odds with the phrase ``provable copyright protection,'' adopted from~\cite{VyasKB23}. We use it in the sense that cryptographers use the term ``provable security'' to describe an analytic approach involving formally defining mathematical properties that limit risk in specific settings and proving that our algorithms satisfy those properties. We can push the analogy further: ``copyright protection'' is, like ``security,'' a broad motivating goal, whereas ``$(\kappa,\beta)$-clean'' is analogous to a specific cryptographic security property like IND-CPA.

\begin{table}[h]
    \small
    \centering
    \begin{tabular}{ll}
    \toprule
    \textbf{Symbol} & \textbf{Definition} \\
    \midrule
    \(\W;~ w,y,z\) & Domain of {works} (e.g., images, text); individual {works} \\
    \(\D\) & Dataset $\D = (w_1, \dots, w_n)$ \\
    \(\train\) & Training algorithm mapping datasets $\D$ to models $p$ \\
    \(p; ~ p(y \mid x)\) & {Conditional generative model}; the probability of generating $y$ on prompt $x$ \\
    \(\aux\) & {Auxiliary information} given to the user (usually $\aux = \ideas(w)$) \\
    \(\u\) & User taking input $\aux$ and access to $p$ and outputting a work $z\gets \u^p(\aux)$  \\
    \(\tau;~\tau(E;\,\aux)\) & {Real-world distribution}; the probability of event $z\in E$ in the \\ & probability experiment $p\gets \train(\D)$, $z \gets \u^p(\aux)$\\
    &\\
    \(\C;~ c\) & Set of \textbf{in-copyright} works; an in-copyright work \\
    \(\subsim(w)\) & Set of works \textbf{substantially similar} to \(w\)\\
    \(\ideas(w)\) & Non-copyrightable \textbf{ideas} associated with $w$ \\
    ``$w'$ stems from $w$'' & $w'$ is a \textbf{derivative work} of $w$, represented as edges in graph $G$ \\ 
    $\C_\D;~ \C_{-\D}$ & In-copyright works from which some work ($\C_\D$) / no work ($\C_{-\D}$) in $\D$ stems \\
    $\scrub(\D,c)$ & Sub-dataset of works in $\D$ that don't stem from $c$ \\
    \(\tau_{-c};~\tau_{-c}(E;\,\aux)\) & Clean-room distribution, defined like $\tau$ but replacing $\D$ with $\scrub(\D,c)$\\
    \bottomrule
    \end{tabular}
    \caption{Common symbols and their definitions. Bolded terms are exogenously (legally) defined.}
    \label{tab:notation}
    \end{table}

\subsection{Notation}
\label{sec:intro:notation}
$\W$ is the domain of \emph{works} protected by copyright law.
For example, $\W$ might be the space of possible images, texts, or songs. We assume $\W$ is discrete.
We usually use $w\in \W$ to denote an individual work, along with $c,y,z\in \W$ as described in the following.
We denote by $\C \sub \W$ the set of all {in-copyright} works---those currently under copyright protection---and denote a particular copyrighted work by $c \in \C$. 
A dataset $\D = (w_1,\dots, w_n)\in \W^*$ is a list of works with multiplicities allowed. We sometimes abuse notation and treat $\D$ as a set.
A \emph{conditional generative model (model)}  $p$ is a mapping from a prompt $x$ to a probability distribution $p(\cdot|x)$ over $\W$ which samples $y\in \W$ with probability $p(y|x)$. 
A training algorithm $\train$ maps a dataset $\D$ to a model $p$.
A user $\u$ is an algorithm which uses black-box access to a model $p$, along with  \emph{auxiliary input} $\aux\in \{0,1\}^*\cup\{\bot\}$, and outputs a work $z\in \W$ (see Section~\ref{sec:interaction}).

\section{The legal-technical interface}
\label{sec:legal-interface}

To use legal concepts within a mathematical formalism, we assume some exogenously-defined functions reflecting and operationalizing those concepts. They are the interface between math and law, enabling mathematical study of legal concepts without formalizing the concepts themselves.
Precise definitions of these concepts are out of scope (perhaps impossible, though see \cite{scheffler2022formalizing} for an attempt at formalizing substantial similarity).
Three legal concepts are needed for this paper: \emph{substantial similarity}, \emph{ideas}, and \emph{access}.
This section defines functions $\subsim$ and $\ideas$ for the first two. Section~\ref{sec:formalizing-access} operationalizes access.

The owner of a copyright in a work has exclusive rights to reproduce, distribute, display, adapt, and perform the work. 
Original works of authorship are eligible for copyright protections as soon as they are fixed in a tangible medium, and eventually expires. 
While the work is protected, it is \emph{in-copyright}.
Many types of works are eligible, including books, music, poetry, plays, choreography, photographs, video, and paintings. 
Copying without permission is sometimes permitted, including under the fair use exception \cite{henderson2023foundation}.

To prove copying, a plaintiff (copyright holder) must prove two things: \emph{substantial similarity} and \emph{access}.\footnote{%
    We do not discuss any sort of secondary liability, including vicarious or contributory infringement.}
The former requires the defendant work's to be substantially similar to the original, with vague tests varying by jurisdiction.\footnote{For example, in the Ninth Circuit it requires determining ``whether the ordinary, reasonable audience would find the works substantially similar in the total concept and feel of the works.''}
Access requires the defendant to have had a reasonable opportunity to view the original (e.g., it's online). Independent creation of substantially similar works is allowed.
The function $\subsim$ defines what it means for one work to be {substantially similar} to another. 
\begin{definition}[$\subsim$]
    For a work $w\in \W$, the set $\subsim(w)\sub \W$ contains all works $w'$ that are substantially similar to $w$. 
    We assume $w$ substantially similar to itself: $w\in \subsim(\W)$. For a set of works $W\sub \W$, we define $\subsim(W) = \cup_{w \in W} \subsim(w)$ as the set of works $w'$ substantially similar to any $w\in W$.
\end{definition}

Copyright does not protect \emph{ideas}, only \emph{original expression}. For example, the text and images in a cookbook may be copyrighted, but not the method of cooking a dish. To be afforded protection, there must be a minimum of creativity beyond the ideas.
The function $\ideas$ defines the non-copyrightable \emph{ideas} of a work, as opposed their expression. We represent these ideas as a string which one can view as encoding a set or any other data type.
\begin{definition}[$\ideas$]
    \label{def:ideas}
    For a work $w\in \W$, the string $\ideas(w)\in \{0,1\}^*\cup\{\bot\}$ is (some representation of) the ideas contained in $w$. For a set of works $W\sub\W$, we define the set of all their ideas $\ideas(W) = \{\ideas(w) : w \in \W\}$.
\end{definition}

We never actually compute $\subsim$ nor $\ideas$. This is a feature: it means we can be agnostic as to the specific instantiation. One might view them as capturing the ``true'' legal concept. Or one may take a more practical view. For example, interpreting $\subsim(w)$ as reflecting the opinion of a typical judge or jury, or viewing $\ideas(w)$ as encoding a set of descriptors.
However the reader wishes to interpret these functions is fine (up to any assumptions stated where used). 
\textbf{In particular, our results hold even for $\subsim(w) = \{w\}$ and $\ideas(w) = \bot$.}

In contrast, a model provider (or data provider) would need to actually work with the copyright dependency graph in Section~\ref{sec:clean-room}. We discuss this limitation in Section~\ref{sec:discussion}.

\section{NAF models do not provide provable copyright protection}
\label{sec:naf}
\emph{Near access-freeness} (NAF) is a mathematical definition for ``provable copyright protection'' \cite{VyasKB23}.
This section describes near access-freeness and its limitations. 
We prove that models can enable verbatim copying while still satisfying NAF. 
While NAF provides protection against a single prompt that is independent of the training data, it makes no guarantees against many prompts (composition)~\cite{li2024va3}, nor a single prompt derived from non-copyrightable {ideas} (not expression).
We emphasize that our focus is \emph{definitions} for provable copyright protection, where NAF falls short. Practically, NAF may still provide useful protection in many natural settings. 
For space, we defer algorithms, proofs, and additional discussion to Appendix~\ref{app:NAF}.
 
\subsection{NAF Background}
\label{sec:naf:background}
\paragraph{NAF informally}
\emph{Near access-freeness} (NAF) is the first (and only other) attempt at a mathematical definition intended to offer ``provable copyright protection'' for generative models~\cite{VyasKB23}.
Motivated by the legal requirement of access, NAF is meant to show ``that the defendant's work is close to a work which was produced without access to the plaintiff's work'' \cite{VyasKB23}.
NAF is closely related to DP. But surprisingly, any training algorithm $\train$ can be used in a black-box way to construct a NAF model without the extra noise that is the hallmark of DP (Theorem~\ref{thm:NAF-cP}). %

The definition of NAF requires a generative model $p$ trained using a copyrighted work $c$ to be close to a ``safe'' model $\safe$ trained without any access to $c$. How close is governed by a parameter $k\ge 0$: smaller $k$ means closer models.
A corollary of the NAF requirement is that for any prompt $x$, the probability $p$ generates something substantially similar to $c$ is at most $2^k$-times greater than $\safe$ doing so. 
We heuristically expect the latter probability to be miniscule, if $\safe$ and $x$ are independent of $c$. If so, this bounds the probability that $p$'s output infringes on prompt $x$.

\paragraph{Known limitations of NAF}
Prior works discuss and critique the NAF framework in three ways: questioning 
NAF's technical effectiveness~\cite{li2024va3}, 
its
legal applicability~\cite{lee2024talkin,elkinkoren2024copyright}, 
and its real-world practicality~\cite{henderson2023foundation,elkinkoren2024copyright,lee2024talkin,chen2024randomization}.

$\ast$ \emph{Technically}, Li, Shen, and Kawaguchi propose and empirically evaluate an attack called VA3 against NAF~\cite{li2024va3}. See Appendix~\ref{app:sec:va3} for a full discussion.
They show that a black-box attacker, can reliably induce a model to produce outputs that infringe on a target work $c^*$. They also give a white-box prompt writing algorithm for diffusion models that greatly improves performance.
VA3 provides good evidence that the $\CPk$ algorithm of \cite{VyasKB23} may not prevent infringement. Some uncertainty remains because the algorithm implemented in VA3 deviates from the original. The paper leaves open whether $k$-NAF (with fixed $k$) prevents copyright infringement, formalizes a flawed attack model, and avoids the underlying definitional questions almost entirely. 

$\ast$ \emph{Legally}, Elkin-Koren et al.\ \cite{elkinkoren2024copyright} give the most extensive response to NAF. They correctly argue that copyright cannot be ``reduced to privacy.'' However, this is a straw man of version of \cite{VyasKB23} and the present paper (see Appendix~\ref{sec:niva}). 
The sharpest criticism is by Lee, Cooper, and Grimmelman who argue that NAF is simply wrong on the law \cite{lee2024talkin}. ``It is not a defense to copyright infringement that you would have copied the work from somewhere else if you hadn't copied it from the plaintiff.'' See Appendix~\ref{app:sec:naf:legal} for more discussion of NAF's legal relevance.

$\ast$ \emph{Practically}, the main concern is that deduplicating data in the way needed to make NAF effective is infeasible \cite{elkinkoren2024copyright,lee2024talkin,henderson2023foundation}. It is unrealistic to produce enough clean training data to pretrain foundation models, advances in deduplication notwithstanding. Moreover, NAF learning algorithms are computationally expensive and may lag in performance~\cite{chen2024randomization}.
These concerns are justified and also apply to DP. 
Still, there may be settings where golden data is practical (Sec.~\ref{sec:discussion}), to say nothing of the value of theory.

\subsection{NAF definition and the CopyProtection ($\CP$) algorithm}
Near access-freeness (NAF) is defined with respect to a function $\safe$.
The $\safe$ function maps a copyrighted data point $c \in \C$ to a generative model $\safe_c\in \Models$ trained without access to $c$. An example $\shardsafe$ is in Appendix~\ref{app:NAF} (Alg.~\ref{alg:shard-safe}).
Thm.~\ref{thm:NAF-cP} presents the main feasibility result:  any training algorithm $\train$ can be used as a black-box to construct an NAF model $\CP$, short for \emph{copy protection}. 

\begin{definition}[Max KL divergence]
    For distributions $p$, $q$, $\dmax(p\|q) := \max_{y \in \Supp(p)} \log \frac{p(y)}{q(y)}.$
\end{definition}

\begin{definition}[$k_x$-NAF  \cite{VyasKB23}]
    Fix a set $\C$ and function $\safe: \C \to \Models$.
    A generative model $p$ is \emph{$k_x$-near access-free ($k_x$-NAF)} on prompt $x \in \cX$ with respect to $\C$ and $\safe$ if for every $c \in \C$,
    $\dmax\bigl(\pdotx ~\|~ \safe_c(\cdot, x) \bigr) \leq k_x.$
A model $p$ is \emph{$k$-NAF} with respect to $\C$ and $\safe$ if for all $x \in \cX$, it is $k_x$-NAF for some $k_x \leq k$. 
\end{definition}
NAF bounds the probability that $p$'s output copies $c$ relative to the probability under $\safe$. If $p$ is $k_x$-NAF safe on prompt $x$ with respect to $\C$ and $\safe$, then for any $c\in \C$:
\begin{equation}
    \label{eq:NAF-similarity-bound}
        p(\subsim(c) ~|~ x) \le 2^{k_x}\cdot \safe_c(\subsim(c) ~|~ x).
\end{equation}
Hence, bounding $k_x$ prevents copying $c$, assuming $\safe_c(\subsim(c) | x)$ is negligibly small (in $|c|$).

\begin{theorem}[CopyProtection algorithm \cite{VyasKB23}]
    \label{thm:NAF-cP}
    Let $p$ be the model returned by $\CP$ (Algorithm~\ref{alg:cp} in Appendix~\ref{app:NAF}), and $q_1$ and $q_2$ be the models returned by $\shardsafe$.  Then $p$ is $k_x$-NAF for $x$ with respect to $\C$ and $\shardsafe$, for some $k_x \le -\log\bigl(1-\dtv\bigl(q_1(\cdot|x), q_2(\cdot|x)\bigr)\bigr)$.
\end{theorem}

\subsection{Failures of NAF models}
\label{sec:naf:failures}
We now describe two ways NAF fails to prevent copying, leveraging its lack of protection against multiple prompts (i.e., composition) and data-dependent prompts. In each case, a user reproduces training data $c\in \D$ verbatim.
In Section~\ref{sec:NAF:CP-counterexample:simplified}, the model generates $c$ when prompted with  $\ideas(c)$---a prompt that depends on $c$ but not on any copyrightable expression thereof. In Section~\ref{sec:NAF:coin-counterexample:simplified}, the user issues a fixed sequence of prompts, recovering $k$ bits of the dataset with each query.

\subsubsection{$\CP$ can regurgitate training data}
\label{sec:NAF:CP-counterexample:simplified}

We show that $\CP$---the main NAF algorithm of \cite{VyasKB23}---can fail to protect against copying. Using a prompt containing no copyrightable expression, a user can cause the NAF model returned by $\CP$ to regurgitate copyrighted training data.\footnote{%
    This doesn't  violate Theorem~\ref{thm:NAF-cP}. 
    We use the fact that the $\CP$ algorithm is not actually $k$-NAF for any fixed $k$. It is $k_x$-NAF, where $k_x$ depends on the prompt $x$.
    Applied to our construction, the bound on $k_x$ given by Theorem~\ref{thm:NAF-cP} is not meaningful for $x \in \ideas(\D)$.
    Adapting our construction to the alternate NAF algorithm $\CPk$ in \cite{VyasKB23} results in a model sometimes fails to terminate.}
This is based on an observation of Thomas Steinke \cite{steinkePersonal}.

Observe that $\CP$ is a black-box transformation from an underlying training algorithm $\train$. Thus, $\CP$ is really a family of algorithms---one per $\train$.
The following theorem states that there exists $\train^*$ such that the resulting NAF algorithm $\CP^*$ fails in the following way. On prompt $x=\ideas(c)$ reproduces training datum $c\in \D$ verbatim, as long as $c$ is uniquely distinguishable by its unprotected ideas.\footnote{For example, $\CP^*$ could output \emph{Oh, the Places You'll Go!} on prompt ``Advice in rhyme for proceeding in life; weathering fear, loneliness, and confusion; and being in charge of your actions.''}

\begin{theorem}
    \label{thm:CP-counter-example:simplified}
    For model training algorithm $\train^*$, let $\CP^*$ be the $k_x$-NAF algorithm from Theorem~\ref{thm:NAF-cP}, and let $p^*\gets \CP^*$. 
    For all $\ideas$, there exists a training algorithm $\train^*$ such that for all datasets $\D$ and for all works $c \in \D$, the following holds. 
    If $c$'s ideas are  distinct in $\D$ (i.e., $\forall w\in \D, w\neq c: \ideas(w)\neq\ideas(c)$), then
        \[p^*\left(c ~|~ \ideas(c)\right) = 1.\]
\end{theorem}

We defer the proof of (a slight generalization) of this theorem to Appendix~\ref{app:sec:naf:counter-examples}.
The proof uses the training algorithm $\train^* = \train_\kv$ (Algorithm~\ref{alg:train-cP-counter-example}) described in Example~\ref{example:tainted-kv}.

\subsubsection{$k$-NAF does not prevent full reconstruction}
\label{sec:NAF:coin-counterexample:simplified}
Our next theorem shows that $k$-NAF may allow training data to be reconstructed, even if $k$ is arbitrarily small and independent of $x$. 
This is because the $k$-NAF guarantee does not compose across a user's many queries, and each query may leak up to $k$ bits of training data. 
In more detail, we give a family of models that are $k$-NAF with respect to pure noise, yet which enable a user to reconstruct the dataset verbatim.
For any $\ell\ge 1$, let $\coin_\ell(\cdot|x)$ be uniform over $\{0,1\}^{\ell}$ for all prompts $x$. As a generative model, $\coin_\ell(\cdot|\cdot)$ is clearly a ``safe'' instantiation of $\safe_c$ for any copyrighted work $c$. 
The proof is deferred to Appendix~\ref{app:sec:naf:counter-examples}.
\begin{theorem}
\label{thm:NAF-counter-example}
Fix $\C \sub \W \sub \{0,1\}^*$. For $\D \in \W^*$, let $L$ be total the length of $\D$ in bits. There exists a (deterministic) training algorithm $\train:(\D,k) \mapsto p_{k,\D}$ satisfying the following.
\begin{itemize}[nosep]
    \item For all $\D$ and $k>0$: $p_{k,\D}$ is $k$-NAF with respect to $\C$ and $\coin_{\ell}$ for $\ell = \max\{1,\lfloor k \rfloor\}$.
    \item There exists a user $\u$ such that for all $\D$ and $k>0$: $\u$ makes $\poly(L,1/k)$ black-box queries to $p_{k,\D}$ and outputs $\D$ with probability $>0.99$.
\end{itemize}
\end{theorem}

\section{The interaction between user and model}
\label{sec:interaction}
People interact with generative models in complex ways: refining prompts and using results to create a final product whose author is not solely human nor machine.
Meaningful copyright protection must cover such cases.
To that end, we make the interaction between a user and the model explicit. 

Let $\u$ be a (randomized) algorithm called the \emph{user}, $p$ be a model, and $\aux\in \{0,1\}^*\cup\{\bot\}$ be an \emph{auxiliary input}.
In cryptography, auxiliary inputs allow one to define security guarantees even when an algorithm has instance-specific side information. In our setting, the algorithm is the user, and the side information represents everything the user knows that is external to the model. Often, we will take $\aux$ to be the non-copyrightable ideas of an in-copyright work: $\aux = \ideas(c)$.

We denote by $\u^p(\aux)$ the algorithm $\u$ run with input $\aux$ and black-box access to $p$. The result is a work $z \in \W$ distributed as $z\gets \u^p(\aux)$.
Together with a training algorithm $\train$ and a dataset $\D$, this process induces probability measure $\tau$ over $\W$, defined next.

\begin{definition}[User's output distribution]
    \label{def:tau}
    For a user $\u$, training algorithm $\train$, dataset $\D$, and auxiliary information $\aux$, we define the \emph{user's output distribution} $\tau$ over $\W$ as:
        \[\tau(w; \aux) = \Pr_{\substack{p \gets \train(\D)\\z \gets \u^p(\aux)}}[z=w].\]
    The probability is taken over the randomness of the algorithms $\train$, $\u$, and $p$.
    Note that $\tau$ depends on $\train$, $\u$, and $\D$. The notation elides these dependencies to reduce clutter.
\end{definition}
Legally, the user's output infringes on the copyright of a work $c\in \C$ only if it is substantially similar to $c$. 
Preventing substantial similarity prevents infringement. 
This motivates our first desideratum.
\begin{desideratum}
    \label{des:minimize-tau}
    We want $\mathbf{\tau(\subsim(\C);\aux)}$ to be as small as possible for as many users as possible.
\end{desideratum}
By considering the user's output distribution, we require protection against multiple prompts (composition) and data-dependent prompts (when $\aux$ contains the ideas of work in the training data, as below). This fixes the shortcomings of NAF that enabled the attacks described in Section~\ref{sec:naf:failures}.

\section{Tainted models enable copying}
\label{sec:tainted}

This section defines what it means for a generative model to be \emph{tainted}. 
Tainted models let users copy training data of whose expression they have no knowledge. Recall from Section~\ref{sec:legal-interface} that copyright protects the original expression in a work $w$, but not its ideas $\ideas(w)$.

We say a training algorithm $\train$ is tainted if there is a fixed algorithm $\u$ that copies work from the training dataset using only trained model and unprotected ideas. 
The order of quantifiers is important: the user $\u$ is fixed while the dataset $\D$ is arbitrary. Such a user cannot possibly be at fault---it cannot know a work in all possible datasets $\D$.

\begin{definition}[Tainted training]\label{def:tainted}
    We say $\train$ is \emph{tainted} with respect to $\ideas$ if there exists a user $\u$ such that for all datasets $\D$ and all works $w \in \D$:
        \begin{equation}
        \tau\bigl(\subsim(\D);~\ideas(w)\bigr) > 0.99.\label{eq:tainted}
        \end{equation}
\end{definition}

If the goal is to prevent copying, being tainted is as bad as it gets. 
Thus, we get our next desideratum. We will see that NAF fails this test but Definition~\ref{def:clean} passes.
\begin{desideratum}\label{des:taint}A meaningful definition of provable copyright protection should bar tainted models.
\end{desideratum}

\subsection{Examples}
As an example, we show that any training algorithm can be made tainted (Example~\ref{example:tainted-kv}). This is used to prove Theorem~\ref{thm:NAF-counter-example}.

\begin{example}
    \label{example:tainted-kv}
    Algorithm~\ref{alg:train-cP-counter-example} describes a tainted algorithm $\train_\kv$, constructed from any training algorithm $\train_0$ in a black-box way. In words, $\train_\kv(\D)$ trains a model $q_0\gets \train_0(\D)$, and also builds a key-value store $I$ mapping ideas $\id$ to the set $D_\id = \{w\in\D : \ideas(w) = \id\}$.
     On prompt $x$, the model $q_\kv$ returns a random element of $I[x]=D_x$ if it is non-empty.  
    Otherwise, it returns a generation sampled from $q_0(\cdot | x)$.
    It is easy to see that $\train_\kv$ is tainted. 
    Consider the user $\u^p(\aux)$ that returns a sample from $p(\cdot | \aux)$.
    Fix $\D$ and $w\in \D$, and let $p\gets \train_\kv(\D)$. By construction, the user always outputs an element of $D_{\ideas(w)}$. Therefore, $\tau(\subsim(\D);\ideas(w)) \ge \tau(\D_{\ideas(w)};\ideas(w)) = 1$. 
\end{example}

As a negative example, the following claim states that any training algorithm that always returns a single fixed model is not tainted, under the mild assumption that it is possible for two works to express the same ideas so differently that no work is similar to both.
    \begin{claim}\label{claim:constant-model}
            Suppose there exists a pair of works $w_0,w_1\in \W$ such that $\ideas(w_0) = \ideas(w_1)$ and $\subsim(w_0) \cap \subsim(w_1) = \emptyset$. Then  for any fixed model $q$, the constant algorithm $\train_q(\D) = q$ is not tainted. 
        \end{claim}
        \begin{proof}
            Let $\D_i = \{w_i\}$ and fix a user $\u$. Let $\tau_i$ be the user's output distribution (Definition~\ref{def:tau}) with $\D = \D_i$ and $\aux = \ideas(w_i)$.
            By construction, $\tau_1 = \tau^* = \tau_2$.
            Note that $\tau^*(w_i)\le 1/2$ for one of $i = 0,1$. Equation \eqref{eq:tainted} is violated for the corresponding $D_i$.
        \end{proof}

\subsection{Tainted NAF models}
Desideratum~\ref{des:taint} lays out a necessary condition for provable copyright protection: exclude tainted training algorithms. NAF fails this test. This is captured by the following corollaries of Theorem~\ref{thm:CP-counter-example} (slightly generalizing Theorem~\ref{thm:CP-counter-example:simplified}) and Theorem~\ref{thm:NAF-counter-example}.

\begin{corollary}[of Thm.~\ref{thm:CP-counter-example}] For any $\ideas$, the NAF algorithm $\CP^*$ given by Thm.~\ref{thm:CP-counter-example:simplified} is tainted.
\end{corollary}
\begin{proof}
    Consider the user $\u^{p}(\aux)$ that returns a sample from $p(\cdot | \aux)$.
    Fix $\D$ and $w\in \D$, and let $p^*\gets \CP^*(\D)$. By Theorem~\ref{thm:CP-counter-example},  $\tau(\subsim(\D);~\ideas(w)) \ge \tau(\D_{\ideas(w)};~\ideas(w)) = 1$. 
\end{proof}

\begin{corollary}[of Thm~\ref{thm:NAF-counter-example}]
    \label{cor:tainted:naf}
    For any $k$, $\ideas$, there exists a tainted algorithm $\train$ such that for all datasets $\D$, the model $p\gets \train(\D)$ is $k$-NAF with respect to $\coin_\ell$.
\end{corollary}
\begin{proof}
Immediate from the statement of Theorem~\ref{thm:NAF-counter-example} 
\end{proof}

\section{Blameless copyright protection: a definitional framework}
\label{sec:blameless}

Our goal is to define provable copyright protection. The previous section gives a negative answer: provable copyright protection should bar tainted models. This section and those that follow give a positive answer: provable copyright protection should protect blameless users. 

We can't hope to guarantee that a generative model never reproduces copyrighted work, even works it was not trained on. Malicious users can always induce copying. (Formally, if the set of in-copyright works $\C$ is non-empty, there exists $\u$ that always infringes: $\tau(\subsim(\C);\bot)=1$.)
Instead, we wish to protect \emph{blameless} users---who don't themselves induce infringement---from unwitting copying.

This suggests a framework for defining meaningful guarantees, which we call \emph{blameless copyright protection}.
First, define a class of blameless users $\B$.
Second, guarantee that for blameless users, the probability of copying $\tau(\subsim(\C);~\aux)$ at most some small $\kappa$.

\begin{definition}[Blameless copyright protection]
    \label{def:framework}
    Fix $\kappa>0$ and a class $\B$ of \emph{blameless users}. We say $\train$ is \emph{$(\kappa,\B)$-copyright protective} if for all  blameless $\u \in \B$ and all $\aux$, $\D\in \W^*$, $\C\sub \W$:
    \[\tau(\subsim(\C);~\aux) < \kappa.\] %
\end{definition}

We place no a priori restriction on the auxiliary information $\aux$, which represents everything the user knows that is external to the model. This could include copyrighted material.
For example, Alice might be writing a children's book and also own a copy of \emph{Green Eggs and Ham}. As long as she's careful, Alice isn't doing anything wrong and should still enjoy some protection.

The missing piece is $\B$: What makes a user blameless? Different answers will yield different versions of blameless copyright protection.
So Definition~\ref{def:framework} is less a definition than a framework for definitions. Instantiating it may require additional assumptions. As we cannot perfectly capture legal blamelessness, we should err on the side of protecting more users rather than less.

In Section~\ref{sec:clean-room}, we offer a first instantiation of blameless copyright protection, inspired by clean-room design.
But there may be very different ways to formalize blamelessness, which we leave for future work. The cryptographic notion of extraction offers an intriguing approach: a user is blameworthy if a copyright-infringing work can be efficiently extracted from the user itself.

\section{Defining clean-room copyright protection}
\label{sec:clean-room}
This section  formalizes what it means to have \emph{access} to a copyrighted work (Section~\ref{sec:formalizing-access}) and  instantiates the blameless copyright protection framework (Section~\ref{sec:clean:distribution})---drawing inspiration from {clean-room design}.
In copyright law, a clean room is ``a process of producing a product
under conditions guaranteeing independent design and foreclosing the
possibility of copying''~\cite{elkins1990nec}. The clean room keeps contamination out, like keeping dust out of a semiconductor lab.
The idea comes from cases involving reverse engineering. Roughly, a team is given a description of design specs of the product (ideas) but not the product itself (expression), and tasked with producing a compatible product.  There is no {access}, as long as the team itself was not spoiled by prior familiarity with the product or its design. As such, the outcome is constructively non-infringing. Even substantial similarities will be the product of independent creation.

Clean-room design is the inspiration for a counterfactual \emph{clean-room distribution} where a user interacts with a model trained without access to certain in-copyright works (Definition~\ref{def:tau-clean}). 
We quantify the \emph{blamelessness} of a user as the probability $\beta$ that the user---in the clean-room distribution---would have produced something that is substantially similar to the excluded works (Definition~\ref{def:blameless}). \emph{Clean-room copyright protection} provides a corresponding bound $\kappa$ on the probability of copying any work in the real distribution (Define~\ref{def:clean}), where $\kappa$ depends on $\beta$.

\subsection{Scrubbing a dataset removes access}
\label{sec:formalizing-access}

A clean room is a setting where a user does not have access to a particular copyrighted work. To pin this down, we first operationalize the legal concept of access.

\paragraph{The copyright dependency graph}
We need a way to say whether a work $w'$ is a derivative of another work $w$ a copyright-relevant way. Such dependencies are directed, with later works stemming from earlier works. 
The \emph{copyright dependency graph} encodes this relation. It is assumed to capture the legal concept of access in the following sense:
\emph{For purposes of copyright law, a dataset $\D$ only gives access to a work $w$ if there is some $w'\in \D$ that stems from $w$}.

\begin{definition}[Copyright dependency graph]
The \emph{copyright dependency graph} is a directed graph $\Graph = (\W, \Edges)$ whose edges reflect dependencies relevant for copyright.
If $(w,w') \in \Edges$, we say that $w'$ \emph{stems from} $w$. We assume that every $w$ {stems} from itself: $(w,w) \in \Edges$ for all $w$.
\end{definition}

It is useful to refer to the set of all copyrighted works for which access is or isn't implied by access to a dataset $\D$. We denote these sets by $\C_\D$ and $\C_{-\D}$, respectively. For the sake of copyright analysis, access to $\D$ does not constitute access to any copyrighted $c\in \C_{-\D}$.
\begin{definition}[$\C_\D$ and $\C_{-\D}$]
    \label{def:CD}
    We define $\C_{\D} = \{c \in \C : \exists w\in \D, (c,w) \in \Edges\}$ as the set of copyrighted works from which any work in $\D$ {stems}. We denote its complement $\C_{-\D} = \C \setminus \C_\D$. 
\end{definition}

\paragraph{The $\scrub$ function}
We define the function $\scrub$ to operationalize the legal concept of {access}. 
It removes from a dataset any work that stems from a target copyrighted work $c$. That is, any $w$ for which $c\to w$ is an edge in the copyright dependency graph.
The result is a new dataset $\scrub(\D,c)\sub \D$. For the sake of copyright analysis, access to $\scrub(\D,c)$ does not constitute to access to $c$.
\begin{definition}[$\scrub$]
    Fix copyright dependency graph $\Graph = (\W, \Edges)$, dataset $\D$ and work $c$. The dataset $\scrub(\D,c) = \left(w \in \D ~:~ (c,w) \not\in \Edges\right)$ is the sub-dataset of $\D$ of works not stemming from $c$.
\end{definition}
Observe that $\C_\D = \{c\in \C: \scrub(\D,c) \neq \D\}$ and $\C_{-\D} = \{c\in \C: \scrub(\D,c) = \D\}$

\subsection{Clean training algorithms: copyright protection for blameless users in a clean room}
\label{sec:clean:distribution}

This section gives a clean-room inspired formulation of copyright protection. Informally, we say that a training algorithm is \emph{$(\kappa,\beta)$-clean} if for all users $\u$: either (i) $\u$ would have copied in a clean-room setting with probability at least $\beta$; or (ii) $\u$ copies in the real world with probability at most $\kappa$. 
Making this precise, we define the user's \emph{clean-room output distribution} (Definition~\ref{def:tau-clean}), \emph{blameless users in the clean room} (Definition~\ref{def:blameless}), and \emph{$(\kappa,\beta)$-clean-room copyright protection} (Definition~\ref{def:clean}).

\subsubsection{The user's clean-room distribution}
We define a user's \emph{clean-room  distribution}, a counterfactual to its true distribution $\tau$.
For a copyrighted work $c$, we denote by $\tau_{-c}$  the user's output distribution in a clean room where the model doesn't depend on $c$.
The only difference from the user's real-world output distribution $\tau$ (Definition~\ref{def:tau}) is that the model is trained on $\scrub(\D,c)$ instead of $\D$.
\begin{definition}[User's clean-room distribution]
    \label{def:tau-clean}
    For user $\u$, training algorithm $\train$, dataset $\D$, work $c$, and auxiliary information $\aux$, we define the user's \emph{clean-room distribution} $\tau_{-c}$ for $c$ as:
    \[\tau_{-c}(w;\aux) = \Pr_{\substack{p \gets \train(\scrub(\D,c))\\z \gets \u^{p}(\aux)}}[z=w].\]
\end{definition}
This is closely related to NAF. The model $p\gets \train(\scrub(\D,c))$ is NAF's safe model $\safe_c$ for an appropriately defined function $\safe$. The clean-room distribution $\tau_{-c}$ is the distribution over outputs that results when the user given $\aux$ interacts with $\safe_c$.

\subsubsection{Clean-room blamelessness and copyright protection}
\label{sec:clean:blameless}
\newcommand{\clean}{\mathsf{clean}}
We postulate that \emph{blameless users can control their risk of copying when using a model trained in a clean room}. 
Given $\beta>0$, a blameless user can guarantee that they copy a work $c$ that was ``outside the clean room'' with probability at most $\beta$.
This should hold even if the user is exposed to $c$ in some other than through the model.
People are constantly exposed to copyrighted works, yet we somehow to manage to avoid copying every day. Using a model trained without access to $c$ shouldn't change that.

We call these users \emph{$\beta$-blameless in a clean room}, or \emph{$\beta$-blameless} for short.
Definition~\ref{def:blameless} formalizes this idea. 
Recall that $\C_{-\D}$ is the set of copyrighted works from which no element of $\D$ stems (Definition~\ref{def:CD}).
Thus, the set of works ``outside the clean room'' is $\C_{-\D}\cup\{c\}$. 

\begin{definition}[Blameless in the clean room ($\beta$-blameless)]
    \label{def:blameless}
    For $0\le \beta \le 1$, a user $\u$ is \emph{$\beta$-blameless in the clean room ($\beta$-blameless)} with respect to $\D$, $\C$, $\train$ if for all $c \in \C$ and all $\aux$:\footnotemark
        \[\tau_{-c}\biggl(\subsim\bigl(\C_{-\D} \cup \{c\}\bigr);~\aux\biggr) \le \beta.\]
    Otherwise, $\u$ is $\beta$-blameworthy.\footnotemark
\end{definition}
\footnotetext{
        Equivalently, $\B^{\clean}_\beta(\D,\C,\ideas) := \{\u : \forall c\in\C, \forall\aux, \tau_{-c}\bigl(\subsim\bigl(\C_{-\D} \cup \{c\}\bigr);~\id\bigr) \le \beta\}.$
        }

Instantiating our framework (Definition~\ref{def:framework}), we now define \emph{clean training algorithms} as those that provide copyright protection to users who are blameless in the clean room.
\begin{definition}[Clean-room copyright protection; $(\kappa,\beta)$-clean]
    \label{def:clean}
    For $\kappa, \beta > 0$, we say $\train$ is \emph{$(\kappa,\beta)$-clean} if for all $\C\sub \W$, $\D\in \W^*$, all users $\u$ who are $\beta$-blameless (w.r.t.\ $\C,\D,\train$), and all $\aux$:
    \[\tau(\subsim(\C);~\aux) \le \kappa.\]
    If this only holds for datasets $\D$ in an admissible set  $\DD\sub \W^*$, we say $\train$ is \emph{$(\kappa,\beta)$-clean} for $\DD$. 
\end{definition}

\begin{remark}[How small is $\beta$?]
    \label{rem:beta}
    We think of $\beta$ as the probability that a ``truly blameless'' user (whatever that means) nevertheless produces something substantially similar to a copyrighted work. There is a  limit to how small this can be. Even a monkey on a typewriter---truly blameless if anyone is---has a non-zero $\beta$.

    Getting a good estimate of achievable $\beta$ is beyond the scope of this work. It should be small enough to be very unlikely to occur by chance: $\beta \ll 1/100$ is certainly doable. It should be large enough to reflect that a relatively small amount of original expression is needed for copyright protection: $\beta \gg 1/10^{100}$ is impossible, even for the monkey.
    We conjecture that $\beta \approx 10^{-6}$ or $10^{-9}$ is achievable by a conscientious user. 
\end{remark}

\subsubsection{Clean-room copyright protection is not tainted}
By Desideratum~\ref{des:taint}, a good definition of provable copyright protection should exclude tainted training algorithms. 
The following theorem shows that clean-room copyright protection passes this test, under a reasonable assumption on $\ideas$ and $\subsim$. Roughly, the assumption is that there exist at least $n$ distinct ideas each of which can be expressed in $m\gg n$ ways satisfying a strong dissimilarity property: that for any distinct $w$ and $w'$, the set of works substantially similar to both $w$ and $w'$ is empty.  
\begin{theorem}
    \label{thm:tainted-unclean}
    Fix $\beta<1$, $n\ge 1$, and $m\ge\frac{n}{\beta}$. 
    Suppose there exists works $W = \left(w_{i,j}\right) \in \W^{m\times n}$ such that for all $i \neq i' \in [m]$ and all $j \in [n]$:
    \[\ideas(w_{i,j}) = \ideas(w_{i',j}) \quad\text{and}\quad \subsim(w_{i,j}) \cap \subsim(w_{i',j}) = \emptyset.\]
    If $\train$ is tainted with respect to $\ideas$, then $\train$ is not $(\kappa,\beta)$-clean.
\end{theorem}
The proof is deferred to Appendix~\ref{app:thm:tainted-unclean}.

\section{Differential privacy's protection against copying}
\label{sec:DP}

Many works have suggested that differential privacy~\cite{dwork2006calibrating} (DP) should protect against copying if the training data are deduplicated~\cite{bousquet2020synthetic,henderson2023foundation,VyasKB23,elkinkoren2024copyright,chen2024randomization,livni2024credit}. This section turns the suggestion into a theorem. For background on DP, see Appendix~\ref{app:DP:prelims}.

At a high level, we show that DP training protects blameless users from copyright infringement (in the sense of Definition~\ref{def:clean}) if the training dataset is \emph{golden}, a copyright-deduplication condition defined below.
If so, the probability of copying is at most $\approx e^\eps\beta N_\D$, where $\beta$ is a user's probability of copying in the clean room and $N_\D = |\C_\D|$ is the number of copyrighted works to which $\D$ grants access. 
To guarantee the risk is at most $\kappa$, the user sets $\beta \approx {\kappa}/{e^\eps N_\D}$. 
The linear dependence on $N_\D$ makes this result of mostly theoretical interest, but it may sometimes be good enough or be further improved under additional assumptions on $\subsim$ and $\ideas$ (Remark~\ref{remark:DP:linear-N}).

This allows a user choose their appetite for copyright risk $\kappa$ and then act to set $\beta$ accordingly. Stephen King writing his next bestseller has a much lower risk tolerance ($\kappa = 10^{-9}$) than Joe Schmoe writing a wedding toast ($\kappa =10^{-1}$). Most users are somewhere in between. For $\eps = 5$ and $N_D = 200,000$ Stephen King needs $\beta \approx 10^{-15}$ while Joe Schmoe only needs $\beta \approx 10^{-6}$. If $\beta \approx 10^{-15}$ is too small (Remark~\ref{rem:beta}),  Stephen King can instead forgo using the LLM altogether.

A dataset $\D$ is \emph{golden} if at most one item $w\in \D$ stems from any copyrighted work $c$.\footnotemark~
That item could be $c$ itself or a derivative work.
No protected original expression appears in more than one element of a golden dataset.
For example, it can contain one copy
or parody of the \emph{Abbey Road} album cover (not both), and many parodies of the out-of-copyright \emph{Mona Lisa}. 
Remark~\ref{remark:golden-data-over-time} explains why golden datasets will typically remain golden as the set of in-copyright works evolves over time.

\begin{definition}[Golden dataset]
    Let $\Graph = (\W,\Edges)$ be a copyright dependency graph, and $\C\sub\W$ be the set of in-copyright works. A dataset $\D$ is \emph{golden} (with respect to $\Graph$, $\C$) if for all $c\in \C$ there is at most one $w\in \D$ such that $(c,w) \in \Edges$. 
\end{definition}
    \footnotetext{%
        We adapt and formalize this idea from \cite{VyasKB23}, where it appears as an informal condition that suffices for their analysis: 
        ``It is important that when we omit a data point $x$ it does not share copyrighted content with many other data points that were included in the training set.''
        We also borrow the term ``golden dataset'', though \cite{VyasKB23} uses it to refer to something entirely different: A dataset that is ``carefully scrutinized to ensure that all material in it is not copyrighted or [is] properly licensed.'' That condition doesn't suffice, as even licensed derivates could cause unlicensed copying by a generative model. The same is true for NAF, despite their suggestion that it would yield a safe model.
    }
\begin{theorem}\label{thm:dp}
    Let $\train$ be $(\eps,\delta)$-differentially private for $\eps>0,\delta\ge0$. Let $N_D = |\C_\D|$. Then $\train$ is $(\kappa,\beta)$-clean for golden datasets $\D$, for all $\beta\ge 0$ and $\kappa \ge (e^\eps N_\D  + 1)\beta + N_\D \delta$.
\end{theorem} 
\begin{proof}
    Fix $\C\sub \W$, $\D$, and user $\u$ that is $\beta$-blameless (Def.~\ref{def:blameless}).
    We must show  that for all $\aux$, $\tau(\subsim(\C);~\aux)\le\kappa$. 
    
    Because $\D$ is golden, $|\D \setminus \scrub(\D,c)| \le 1$. Hence, datasets $\scrub(\D,c)$ and $\D$ are either equal or neighboring for all $c\in \C$.
    Applying DP and post-processing, we have $\tau(E) \le e^\eps\cdot \tau_{-c}(E) + \delta$ for all events $E$ and all $c\in \C$.
\begin{align*}
    \tau(\subsim(\C);~\aux) 
    &\le 
    \tau(\subsim(\C_{-\D});~\aux)+\sum_{c \in \C_\D} \tau(\subsim(c);~\aux) &\text{(union bound)}\\
    &\le \beta + N_\D\delta + 
    e^\eps \cdot \sum_{c \in \C_\D} \tau_{-c}(\subsim(c);~\aux)  &\text{(DP; Prop.~\ref{prop:blameless-real-world} below)}\\
    &\le
    (e^\eps N_\D  + 1)\beta + N_\D \delta &\text{($\beta$-blameless)} \\
    &\le \kappa  &\qedhere 
\end{align*}
\end{proof}

In the proof above, we need to bound the probability $\tau(\subsim(\C_{-\D});~\aux)$ of similarity with some work in $\C_{-\D}$ in the real distribution $\tau$. Blamelessness only gives us a bound in the clean-room distribution $\tau_{-c}$. Proposition~\ref{prop:blameless-real-world} shows that the latter implies the former.

\begin{proposition}
    \label{prop:blameless-real-world}
Let $\u$ be $\beta$-blameless w.r.t.\ $\D$, $\C$, $\train$. For all $\aux$:  $\tau\left(\C_{-\D};~\aux\right) \le \beta.$
\end{proposition}
\begin{proof}
By definition of $\C_{-\D}$, $\scrub(\D,c) = \D$ and hence $\tau_{-c} = \tau$. Also, $\C_{-\D} =\C_{-\D}\cup \{c\}$. Thus, $\tau(\C_{-\D}; \aux) \le \tau_{-c}\left(\subsim\left(\C_{-\D} \cup \{c\}\right);\aux\right) \le \beta$, with the last inequality by blamelessness.
\end{proof}

\begin{remark}[On the linear dependence on $N_D$]\label{remark:DP:linear-N}
    Notice that $\kappa > N_D\cdot(\beta + \delta)$ grows linearly with $N_D = |\C_\D|$. This only yields a meaningful bound if $N_D \ll 1/(\beta+\delta)$. Along with the difficulty of creating enormous golden datasets (Section~\ref{sec:discussion}), this is another reason that our approach is most practical when only medium-sized datasets are required (i.e., tens of thousands of items)

    On the other hand, the factor of $N_D = |\C_\D|$ from the union bound is unreasonably pessimistic, and it should be possible to greatly improve it by imposing additional assumptions on $\subsim$ and $\ideas$. Ideas sometimes differ so greatly that the risk of substantial similarity is essentially 0: a photorealistic bird will never resemble Dr.\ Seuss's {Cat in the Hat}. If the user is targeting a fixed set of ideas $\id$ (conditioning $\tau$ on $\id$), one might instead take the union bound over $\C_\D^\id = \{c \in \C_\D : \id \in \ideas(\subsim(c))\}$. We expect that typically $N_D^\id := |\C_\D^\id| \ll |\C_\D| = N_D$. 
\end{remark}

    \begin{remark}[Golden datasets remain golden over time]\label{remark:golden-data-over-time}
        Theorem~\ref{thm:dp} states that $(\eps,\delta)$-differentially private models are $(\kappa,\beta)$-clean for \emph{golden datasets} (and appropriate $\kappa$ and $\beta$). The set of golden datasets $\DD$ depends on the set of in-copyright works $\C$. For example, every $\D$ is golden when nothing is in-copyright ($\C = \emptyset$).
        
        This presents a challenge: Will a golden dataset remain golden as the set of works protected by copyright evolves? As copyrights expire and new works are created, the sets of in-copyright works today $\C$ and tomorrow $\C'$ will diverge. For Definition~\ref{def:clean} to offer lasting protection, we need $\DD \sub \DD'$.

        Fortunately, it will typically be true that datasets remain golden as the set of in-copyright works evolves.
        This requires three observations.  
        First, expiring copyrights don't affect goldenness. If $\D$ is golden for $\C$, then $\D$ is golden for $\C \cap \C'$.
        Second, a work is afforded copyright protections as soon as it is fixed in a tangible medium. This means that if a work currently exists will ever be in-copyright (in $\C'$), then it already in-copyright (in $\C$). Hence, $\C'\setminus \C$ contains only works that do not yet exist. The exception is when a work becomes newly eligible for copyright protections, and the change in eligibility is applied retroactively. For example, if a court extends copyright to a new class of works.
        Third, no work can stem from a work created later. Combined with the previous observation, $\D$ is golden with respect to $\C'\setminus \C$. 
        Putting it all together, we have that $\D \in \DD$ implies that $\D$ is golden with respect to $(\C \cap \C') \cup (\C' \setminus \C) = \C'$. That is, $\D \in \DD'$.
    \end{remark}

\section{Discussion}
\label{sec:discussion}

\subsection{An indemnification policy for genAI outputs}
People copy without generative models, and will continue copying with them.
One way to evaluate the usefulness of a copyright-mitigation measure is to consider what happens when copying does occur. Who should pay? 

We are not asking where the liability would lie under existing law, but where it should lie. When generative AI is involved, holding a user strictly liable for all infringement is neither just nor effective. Punishment cannot incentivize users to avoid copying when avoiding copying is beyond their control. %

It is hard to assign culpability on the basis of NAF alone, as tainted NAF models can induce a user acting appropriately to infringe.
Clean-room copyright protection gives a clearer theoretical account of who is culpable if training is differentially private.  The provider is culpable if the data wasn't golden; the user is culpable if not blameless; possibly both are culpable, or neither. However, this test is impossible to apply as blamelessness cannot in general be checked.

Even so, a simple indemnification policy is possible: {The model provider pays for infringement if the copied work appears too often in the training data.} This is a question of fact that legal process and forensic experts are well-suited to resolve. Experts for each party can examine the training data and testify before the finder of fact. 

This indemnification policy gives users what they want. Users are indemnified when there is any possibility that the model is to blame, and they can control their risk tolerance by tuning $\beta$.
The policy also gives model providers what they want. Besides offering a valuable protection for users, the provider can trade off their overall exposure to copyright penalties with their effort spent cleaning the data. Perhaps the provider can even pass the liability to third-party data providers contracted to provide golden data.

Many generative AI services already indemnify their users against copyright liability, with restrictions on user conduct. For example, OpenAI's terms do not apply where the ``Output was modified, transformed, or used in combination with products or services not provided by or on behalf of OpenAI'' or where the user ``did not have the right to use the Input or fine-tuning files to generate the allegedly infringing Output.''\footnote{\url{https://openai.com/policies/service-terms/} (accessed June 4, 2025).} These are sensible restrictions. However, deciding whether the indemnification applies could require a complex forensic reconstruction and analysis of the creation of the infringing artifact, including how third-party tools were used and the role of the user's prompts or fine-tuning. In contrast, our indemnification policy imposes no restrictions on the user and reduces the question of liability to a clear question of fact (albeit with the significant drawbacks of requiring golden datasets and differential privacy).

\subsection{Is clean room protection practical?}
Theorem~\ref{thm:dp} requires DP models, golden data, and blameless users. Is it practical?

Making a golden dataset is much harder than deduplication (already a major challenge~\cite{lee2021deduplicating}): it depends on squishy legal standards and on data outside the dataset. 
Still, creating a golden dataset with tens of thousands of items seems feasible but expensive. One only needs to be able to determine if two items stem from a common in-copyright work (i.e., siblings in the copyright dependency graph). For a given pair of items, doing so with confidence is plausible with appropriate domain expertise.
Another approach would be to create or commission new work with known copyright dependencies. 

Constructing a golden dataset large enough to pre-train a foundation model is hopeless. Applications with more modest data needs are more promising, fine-tuning a pre-trained model for example. Even with differential privacy, tens of thousands of training data can achieve good utility.

There are non-technical ways to address the challenge of golden datasets. Suppose the trainer license the data from a data provider. They could require metadata to include copyright dependencies, or that the data provider provide golden data. Either way, the license agreement can indemnify users if the data provider's failure to appropriately clean or tag the data leads to infringement.
Not all is lost if the data is imperfect. The guarantees will still hold for any work satisfying the golden condition.

Blamelessness presents a greater challenge: it is not checkable, not even by the user.
Still, we believe that diligent users who are attentive to the possibility of inadvertent infringement can guarantee $\beta$-blamelessness for $\beta$ small enough to make Theorem~\ref{thm:dp} meaningful. 

This is the crux of the matter. Why do I think that users can avoid being unduly influenced by works to which they have been exposed? I don't have a good answer.
More than anything, I can't shake the belief that I could do it. That I could productively use ChatGPT-4o to produce a story or image that is wholly original, bearing no resemblance even to stories and images with which I am intimately familiar.
People are able to be truly creative, despite constant exposure to copyrighted work. 
Somehow you and I manage to avoid copying every day.

\section*{Acknowledgements}
    We are indebted to Mayank Varia for many helpful discussions and encouragement. We thank Gautam Kamath, Yu-Xiang Wang, Seewong Oh, and especially Thomas Steinke for early discussions when the idea of this paper was still taking shape. Theorem~\ref{thm:CP-counter-example:simplified} is based on an observation of Thomas. We thank Sarah Scheffler, Randy Picker, Lior Strahilevitz, James Grimmelmann, anonymous reviewers, and attendees of the 2024 Works-in-Progress Roundtable on Law and Computer Science at the University of Pennsylvania for feedback on an early draft.

\bibliographystyle{alpha}
\bibliography{refs}

\newcommand{\etalchar}[1]{$^{#1}$}
\begin{thebibliography}{EKHLM24}

\bibitem[BLM20]{bousquet2020synthetic}
Olivier Bousquet, Roi Livni, and Shay Moran.
\newblock Synthetic data generators--sequential and private.
\newblock {\em Advances in Neural Information Processing Systems},
  33:7114--7124, 2020.

\bibitem[CF19]{chatterjee2019minds}
Mala Chatterjee and Jeanne~C Fromer.
\newblock Minds, machines, and the law.
\newblock {\em Columbia Law Review}, 119(7):1887--1916, 2019.

\bibitem[CKOX24]{chen2024randomization}
Wei-Ning Chen, Peter Kairouz, Sewoong Oh, and Zheng Xu.
\newblock Randomization techniques to mitigate the risk of copyright
  infringement.
\newblock {\em arXiv preprint arXiv:2408.13278}, 2024.

\bibitem[CLG{\etalchar{+}}23]{cooper2023report}
A~Feder Cooper, Katherine Lee, James Grimmelmann, Daphne Ippolito, Christopher
  Callison-Burch, Christopher~A Choquette-Choo, Niloofar Mireshghallah, Miles
  Brundage, David Mimno, Madiha~Zahrah Choksi, et~al.
\newblock Report of the 1st workshop on generative ai and law.
\newblock {\em arXiv preprint arXiv:2311.06477}, 2023.

\bibitem[DMNS06]{dwork2006calibrating}
Cynthia Dwork, Frank McSherry, Kobbi Nissim, and Adam Smith.
\newblock Calibrating noise to sensitivity in private data analysis.
\newblock In {\em Theory of Cryptography: Third Theory of Cryptography
  Conference, TCC 2006, New York, NY, USA, March 4-7, 2006. Proceedings 3},
  pages 265--284. Springer, 2006.

\bibitem[EKHLM24]{elkinkoren2024copyright}
Niva Elkin-Koren, Uri Hacohen, Roi Livni, and Shay Moran.
\newblock Can copyright be reduced to privacy?
\newblock Forthcoming, Foundations of Responsible Computing, 2024.
\newblock \url{https://arxiv.org/abs/2305.14822}.

\bibitem[Elk90]{elkins1990nec}
David~S Elkins.
\newblock {NEC v.\ Intel}: A guide to using ``clean room'' procedures as
  evidence.
\newblock {\em 10 {Computer L.J.\ 453}}, 1990.

\bibitem[GAZ{\etalchar{+}}24]{golatkar2024cpr}
Aditya Golatkar, Alessandro Achille, Luca Zancato, Yu-Xiang Wang, Ashwin
  Swaminathan, and Stefano Soatto.
\newblock Cpr: Retrieval augmented generation for copyright protection.
\newblock In {\em Proceedings of the IEEE/CVF Conference on Computer Vision and
  Pattern Recognition}, pages 12374--12384, 2024.

\bibitem[Goo25]{goodyear2025artificial}
Michael~P Goodyear.
\newblock Artificial infringement.
\newblock In {\em Proceedings of the 2025 Symposium on Computer Science and
  Law}, pages 26--38, 2025.

\bibitem[HLJ{\etalchar{+}}23]{henderson2023foundation}
Peter Henderson, Xuechen Li, Dan Jurafsky, Tatsunori Hashimoto, Mark~A. Lemley,
  and Percy Liang.
\newblock Foundation models and fair use.
\newblock {\em Journal of Machine Learning Research}, 24(400):1--79, 2023.

\bibitem[LCG24]{lee2024talkin}
Katherine Lee, A.~Feder Cooper, and James Grimmelmann.
\newblock Talkin' 'bout ai generation: Copyright and the generative-ai supply
  chain.
\newblock Forthcoming, Journal of the Copyright Society, 2024.
\newblock \url{https://papers.ssrn.com/sol3/papers.cfm?abstract_id=4523551}.

\bibitem[LIN{\etalchar{+}}21]{lee2021deduplicating}
Katherine Lee, Daphne Ippolito, Andrew Nystrom, Chiyuan Zhang, Douglas Eck,
  Chris Callison-Burch, and Nicholas Carlini.
\newblock Deduplicating training data makes language models better.
\newblock {\em arXiv preprint arXiv:2107.06499}, 2021.

\bibitem[LMNP24]{livni2024credit}
Roi Livni, Shay Moran, Kobbi Nissim, and Chirag Pabbaraju.
\newblock Credit attribution and stable compression.
\newblock {\em Advances in Neural Information Processing Systems},
  37:2663--2685, 2024.

\bibitem[LSK24]{li2024va3}
Xiang Li, Qianli Shen, and Kenji Kawaguchi.
\newblock Va3: Virtually assured amplification attack on probabilistic
  copyright protection for text-to-image generative models.
\newblock In {\em Proceedings of the IEEE/CVF Conference on Computer Vision and
  Pattern Recognition}, pages 12363--12373, 2024.

\bibitem[Ste23]{steinkePersonal}
Thomas Steinke.
\newblock Personal communication, 2023.

\bibitem[STV22]{scheffler2022formalizing}
Sarah Scheffler, Eran Tromer, and Mayank Varia.
\newblock Formalizing human ingenuity: A quantitative framework for copyright
  law's substantial similarity.
\newblock In {\em Proceedings of the 2022 Symposium on Computer Science and
  Law}, pages 37--49, 2022.

\bibitem[VKB23]{VyasKB23}
Nikhil Vyas, Sham Kakade, and Boaz Barak.
\newblock Provable copyright protection for generative models.
\newblock {\em arXiv preprint arXiv:2302.10870}, 2023.

\end{thebibliography}

\appendix

\section{Are we trying to reduce copyright to privacy? No!}
\label{sec:niva}    
Elkin-Koren, Hacohen, Livni, and Moran argue that copyright cannot be ``reduced to privacy.''  Referring to both NAF and DP by umbrella term algorithmic stability, they argue that the sort of provable guarantees that \cite{VyasKB23} and this paper seek do not capture copyright's complexities.
\begin{quote}
Algorithmic stability approaches, when used to establish proof of copyright infringement are either too strict or too lenient 
from a legal perspective. Due to this misfit, applying algorithmic stability approaches as filters for generative models will likely to distort the delicate balance that copyright law aims to achieve between economic incentives and access to creative works.
\end{quote}
Too strict by excluding permitted uses of copyrighted data: (1) works in the public domain; (2) unprotected aspects of copyrighted work (e.g., ideas, facts, procedures); and (3) lawful uses of copyrighted work, especially fair use.
Too lenient when protected expression originating in one work is present in many other works in a training data set. For example, copies, derivatives, or snippets of the original (whether fair use or not) would undermine any NAF- or DP-based guarantee if unaccounted for. 

We completely agree, and suspect that \cite{VyasKB23} would too. The claim that ``copyright can be reduced to privacy'' is a straw man. It would be foolish to suggest that the whole of copyright law for generative AI be governed my a mathematical formalism like NAF or DP. That doesn't mean that a mathematical formalism offering legal guarantees isn't worthwhile---as a thought experiment or as a first step towards practical solutions.

Still, a formalism that is both too lenient and too strict won't support a legal conclusion either way. As we cannot reduce copyright to a mathematical formalism, we must choose how to err. In this work, we seek sufficient conditions for preventing infringement. Conditions that are too strict will leave room for improvement, but won't be fatally flawed.  But conditions that are too lenient would be fatal---they would not suffice. Thus, we must address the possibility that one work's copyrighted expression appears in many other works.

\section{Deferred material on near access-free models (Section~\ref{sec:naf})}
\label{app:NAF}

This section presents our detailed treatment of near access-free models.
It describes near access-freeness and its limitations. 
We prove that models can enable verbatim copying while still satisfying NAF. 
While NAF provides protection against a single prompt that is independent of the training data, it makes no guarantees against many prompts~\cite{li2024va3}, nor a single prompt derived from non-copyrightable {ideas} (not expression).

To make this section self-contained, we repeat material from Section~\ref{sec:naf}, quoting freely.

\subsection{Definitions and main results from \cite{VyasKB23}}
\label{sec:naf:def}
NAF is defined with respect to a function $\safe$.  
The $\safe$ function maps a copyrighted data point $c \in \C$ to a generative model $\safe_c\in \Models$ trained without access to $c$. 
An example $\safe$ function is $\shardsafe$.
We simplify the notation of \cite{VyasKB23} by fixing the divergence to maximum KL divergence.
\begin{definition}[Max KL divergence]
    For distributions $p$, $q$, $\dmax(p\|q) := \max_{y \in \Supp(p)} \log \frac{p(y)}{q(y)}.$
\end{definition}

\begin{definition}[$k_x$-NAF  \cite{VyasKB23}]
    Fix a set $\C$ and function $\safe: \C \to \Models$.
    A generative model $p$ is \emph{$k_x$-near access-free ($k_x$-NAF)} on prompt $x \in \cX$ with respect to $\C$ and $\safe$ if for every $c \in \C$,
\[
\dmax\biggl(\pdotx ~\|~ \safe_c(\cdot, x) \biggr) \leq k_x.
\]
A model $p$ is \emph{$k$-NAF} with respect to $\C$ and $\safe$ if for all $x \in \cX$, it is $k_x$-NAF for some $k_x \leq k$. 
\end{definition}

The appeal of this definition is that it can be used to bound the probability that model $p$ produces outputs that violate the copyright of a work $c$, relative to the probability under $\safe$. 
\begin{lemma}[$k$-NAF event bound \cite{VyasKB23}]\label{lemma:NAF-event-bound}
Suppose model $p$ is $k_x$-NAF on prompt $x$ with respect to $\C$ and $\safe$. Then for any $c \in \C$ and any event $E \sub \cY$:
\[
p(E|x) \le 2^{k_x}\cdot \safe_c(E|x).
\]
\end{lemma}

\noindent
Letting $E=\subsim(c)$ be the event that $y$ is substantially similar to $c$, we get \[p(\subsim(c) | x) \le 2^{k_x} \cdot \safe_c(\subsim(c)|x).\]
As copying requires substantial similarity, bounding $k_x$ suffices to prevent violation whenever $\safe_c(\subsim(c) | x)$ is negligibly small. Heuristically, we expect $\safe_c(\subsim(c) | x)$ to be negligible in $|c|$. It may sometimes be large, e.g., when $x$ contains the copyrighted work $c$ \cite{VyasKB23}.

Theorem~\ref{thm:NAF-cP} presents the main feasibility result of \cite{VyasKB23}:  any training algorithm $\train$ can be used as a black-box to construct an NAF model $\CP$, short for \emph{copy protection}. 
\begin{theorem}[\cite{VyasKB23}]
    Let $p$ be the model returned by $\CP$, and $q_1$ and $q_2$ be the models returned by $\shardsafe$. Then $p$ is $k_x$-NAF with respect to $\C$ and $\shardsafe$, with
    \begin{equation}
        k_x \le -\log\biggl(1-\dtv\bigl(q_1(\cdot|x), q_2(\cdot|x)\bigr)\biggr).
    \end{equation}
\end{theorem}
Observe that $\CP(\D)$ is well-defined if and only if $\forall x$, $\exists y$: $q_1(y|x) > 0$ and $q_2(y|x) > 0$.

\begin{algorithm}
    \caption{$\shardsafe$ \cite{VyasKB23}}\label{alg:shard-safe}
    \KwParams{Dataset $\D$, training algorithm $\train$}
    Do the following once: \tcp*[h]{or derandomize for statelessness}\;
        \Indp
        Partition $\D$ into disjoint datasets $\D_1$ and $\D_2$\;
        Set $q_1 \gets \train(\D_1)$, $q_2 \gets \train(\D_2)$\;
        \Indm
    \BlankLine
    \KwInput{$c \in \C$}
    Let $i = \min \{j : c \not \in D_j\}$\;
    \KwResult{$q_i$}
    \end{algorithm}

\begin{algorithm}
    \caption{$\CP$: Copy-Protection \cite{VyasKB23}}\label{alg:cp}
    \KwInput{Dataset $\D$}
    \KwLearn{Run $\shardsafe(\D)$ to obtain $q_1$, $q_2$ as in Algorithm~\ref{alg:shard-safe}}
    \KwResult{The model $p$ with 
    \[
    p(y|x) = \frac{\min\{q_1(y|x),q_2(y|x)\}}{Z(x)}
    \]
    where $Z(x)$ is a normalization constant that depends on $x$.
    }
\end{algorithm}

\subsection{Failures of NAF models}
\label{app:sec:naf:counter-examples}

\subsubsection{$\CP$ can regurgitate training data}
We show that $\CP$---the main NAF algorithm of \cite{VyasKB23}---can fail to protect against copying. Using a prompt containing no copyrightable expression, a user can cause the NAF model returned by $\CP$ to regurgitate copyrighted training data.
This is based on an observation of Thomas Steinke \cite{steinkePersonal}.

The following claim is a slight generalization of Theorem~\ref{thm:CP-counter-example:simplified}. In words, instantiating $\CP=\CP_\kv$ with the (tainted) training algorithm $\train_\kv$ (Algorithm~\ref{alg:train-cP-counter-example}) described in Example~\ref{example:tainted-kv}, causes training data regurgitation.

Note that $\train_\kv$ is itself built from some underlying training algorithm $\train_0$. The only requirement we need of $\train_0$ is that it produces models with \emph{full support}. That is, for every dataset $\D$, model $p_0\gets \train_0(\D)$, prompt $x$, and output $y \in \W$, we require $p(y|x) > 0$.

\begin{theorem}
    \label{thm:CP-counter-example}
Let $\D$ be a dataset. For ideas $\id \in \Ideas$, let $\D_\id = \{w' \in \D : \ideas(w') = \id\}$. 
Let $\train_0$ produce models with full support. 
Let $\train_\kv$, $\shardsafe_\kv$ and $\CP_\kv$ be as defined in Algorithms~\ref{alg:train-cP-counter-example}, \ref{alg:shard-safe}, and \ref{alg:cp}, defined with respect to $\train_0$ (and the preceding algorithms).
Let $p \gets \CP_\kv(\D)$. Then  for all $w\in\D$: \[p\left(\D_{\ideas(w)} ~|~ \ideas(w)\right) = 1.\]
\end{theorem}

\noindent
In particular, for any copyrighted work $c\in \D$ whose ideas are unique, we get \[p(c ~|~ \ideas(c)) = 1.\]
Theorem~\ref{thm:CP-counter-example} does not contradict the $\CP$ theorem (Theorem~\ref{thm:NAF-cP}). This is because $\CP$ is only $k_x$-NAF, where $k_x$ depends on the prompt $x$. In our construction, the bound on $k_x$ given by Theorem~\ref{thm:NAF-cP} is not vacuous for $x \in \ideas(\D)$.

\begin{proof}[Proof of Theorem~\ref{thm:CP-counter-example}]
Unrolling the algorithms, $\CP_\kv(\D)$ calls $\shardsafe_\kv$, which shards the data into $\D_1$ and $\D_2$ and trains $q_i \gets \train_\kv(\D_i)$. Without loss of generality, suppose $w \in \D_1$. Let $\id = \ideas(w)$. 

By construction, $q_1(\D_{\id} | \id) = 1$ (as in Example~\ref{example:tainted-kv}). Thus, for all outputs $y\in \W$:
\[
p(y|\id) \propto \min\bigl\{q_1(y | \id), q_2(y | \id)\bigr\} 
= 
\begin{cases} 
    q_2(y | \id) & \text{if } y \in  \D_{\id } \\ 
    0 & \text{if } y\not \in \D_{\id} \end{cases}
\]
The distribution $p(\cdot | \id)$ is well-defined if there exists $y\in \D_{\id}$ such that $q_2(y|\id) > 0$. The model $q_2( \cdot | \id)$ returns an element of $\D_2 \cap \D_{\id}$ if it has non-empty intersection, or it returns a sample from an underlying model with full support (the model returned by $\train_0$). Either way, $q_2(y | \id) > 0$ for all $y \in \D_2 \cap \D_{\id}$.

The claim follows immediately:
\[p(\D_{\id} ~|~ \id) = 1 - p(\overline{\D}_{\id} ~|~ \id) = 1.\qedhere\]
\end{proof}

\begin{algorithm}
    \caption{$\train_\kv$ (for Example~\ref{example:tainted-kv} and Theorem~\ref{thm:CP-counter-example})}\label{alg:train-cP-counter-example}
    \KwParams{Training algorithm $\train_0$}
    \KwInput{Data $\D$}
    \KwOutput{Model $q_\kv$}
    \BlankLine
    Let $q_0 \gets \train_0(\D)$\;
    Initialize empty key-value store $I$, whose keys are ideas $\id$, values are sets of works $W \subset \W$\;
    \For{$w \in \D$}{
        $I[\ideas(w)] \gets I[\ideas(w)] \cup \{w\}$ \quad \texttt{// add w to I[ideas(w)]}\;
    }
    \BlankLine
    Let $q_\kv$ the conditional generative model which on prompt $x$ does the following:\;
    \Indp
    $W \gets I[x]$\;
    \lIf{$W \neq \emptyset$}{return $y$ sampled uniformly from $W$}
    \lElse{return $y$ sampled from $q_0(\cdot|x)$}
    \Indm
    \BlankLine
    \KwResult{$q_\kv$}
\end{algorithm}

\subsubsection{$k$-NAF does not prevent full reconstruction}
Our next theorem shows that $k$-NAF may allow training data to be reconstructed, even if $k$ is arbitrarily small and independent of $x$. 
This is because the $k$-NAF guarantee does not compose across a user's many queries, and each query may leak up to $k$ bits of training data. 

We give a family of models that are $k$-NAF with respect to a model that returns pure noise, yet enable a user to reconstruct the dataset verbatim.
For any $\ell\ge 1$, let $\coin_\ell(\cdot|x)$ be uniform over $\{0,1\}^{\ell}$ for all prompts $x$. As a generative model, $\coin_\ell(\cdot|\cdot)$ is clearly a ``safe'' instantiation of $\safe_c$ for any copyrighted work $c$. 
\begin{theorem*}[Theorem~\ref{thm:NAF-counter-example}]
    Fix $\C \sub \W \sub \{0,1\}^*$. For $\D \in \W^*$, let $L$ be total the length of $\D$ in bits. There exists a (deterministic) training algorithm $\train:(\D,k) \mapsto p_{k,\D}$ satisfying the following.
    \begin{itemize}
        \item For all $\D$ and $k>0$: $p_{k,\D}$ is $k$-NAF with respect to $\C$ and $\coin_{\ell}$ for $\ell = \max\{1,\lfloor k \rfloor\}$.
        \item There exists a user $\u$ such that for all $\D$ and $k>0$: $\u$ makes $\poly(L,1/k)$ queries to $p_{k,\D}$ and outputs $\D$ with probability $>0.99$.
    \end{itemize}
    \end{theorem*}
    
    \begin{remark}
        The theorem and proof can be adapted to more realistic $\safe$ by encoding $\D$ in the bias of a hash of $p$'s outputs. 
        But the added complexity would obscure the technical idea used in the proof: biasing $\safe$ can reveal $\D$ without violating NAF. 
    \end{remark}
    
    \begin{proof}[Proof of Theorem~\ref{thm:NAF-counter-example}]
        We parse $\D$ as a bit string of length $L$, and let $\D[j]$ be its $j$th bit. We prove the result separately for $k\ge 1$ and $0<k<1$.
        
    For $k\ge 1$, $\train$ outputs the model $p_{k,\D}$ as follows:
    \[
        p_{k,\D}(\cdot|x) = \begin{cases}
            (\D[x],\D[x+1], \dots, \D[x+\ell-1]) & \text{if }x \in [L-\ell+1]\\
            \text{sample uniform } y\in\{0,1\}^{\ell}& \text{otherwise}
        \end{cases}
        \]
    The model $p_{k,\D}$ is $k$-NAF with respect to $\C$ and $\coin_\ell$: 
    $\forall x$, $\dmax\bigl(p_{k,\D}(\cdot|x) ~\|~ \safe_c(\cdot, x) \bigr) \le  k.$
    To reconstruct $\D$, the user $\u$ queries $p_\ell(\cdot|x)$ for $x = i\ell+1$ for $i = 0,1,\dots, (L-1)/\ell$.
    
    For $0<k<1$, $\train$ sets $\beta = 2^k-1$ and outputs the model $p_{k,\D}$ as follows:
    \[
        p_{k,\D}(\cdot|x) = \begin{cases}
            y\sim \Bern(\frac{1}{2} + \beta(\D[x] - \frac{1}{2})) & \text{if }x \in [L]\\
            y\sim \Bern(\frac{1}{2}) & \text{otherwise}
        \end{cases}
        \]
    The intuition is that $p_{k,\D}$ encodes $\D[x]$ in the bias of the output of $p_{k,\D}(\cdot|x)$, with the magnitude of $\beta \in (0,1)$ controlling the strength of the bias. The model $p_{k,\D}$ is $k$-NAF with respect to $\coin_1$: $\forall x$, $\dmax\bigl(p_{k,\D}(\cdot|x) ~\|~ \safe_c(\cdot, x) \bigr) = \log \frac{1/2 + \beta(\D[x]-1/2)}{1/2} \le \log (1 + \beta) \le k.$
    
    To determine $\D[j]$ with greater than $>1-\frac{1}{100L}$, the user $\u$ makes $\poly(L,\log 1/\beta) = \poly(L,1/k)$ queries to the model, and stores the majority.
    The user does this for each $j \in [L]$ and outputs the result.
    By a union bound, the user's output is equal to $\D$ with probability greater than $0.99$.
    \end{proof}

\subsection{\emph{VA3}: an empirical attack on NAF \cite{li2024va3}} 
\label{app:sec:va3}
Li, Shen, and Kawaguchi propose and empirically evaluate a ``Virtually Assured Amplification Attack'' against NAF~\cite{li2024va3}. 
At a high level, \cite{li2024va3} shows that an attacker, interacting with an NAF model in a black-box manner, can reliably induce a model to produce outputs that infringe on a target work $c^*$. They also give a white-box prompt writing algorithm for diffusion models (Anti-NAF) that greatly improves the performance of their attacks.
    
Overall, the work provides good evidence that the algorithms proposed by \cite{VyasKB23} may not prevent infringement. Some uncertainty remains because the algorithm implemented in \cite{li2024va3} deviates from the original, as explained below. The paper leaves open whether $k$-NAF (with fixed $k$) prevents copyright infringement, formalizes a flawed attack model, and avoids the underlying definitional questions almost entirely.

We now describe \cite{li2024va3} in more detail, offering a somewhat different interpretation than the original.

The paper's main focus is the Amplification Attack. Given a target copyrighted work $c^*$, an attacker repeatedly queries a model with some prompt $x=x(c^*)$. It returns the generation most similar to $c^*$. This process amplifies the one-shot probability of infringement. For example, from $0.40\%$ to $13.64\%$ when $x$ is the original caption of $c^*$ in the training data, or from $8.52\%$ to $77.36\%$ using the Anti-NAF prompt generator.
In our view, this not actually the paper's main negative result for NAF. 

More significant (for our purposes) is the existence of prompts $x$ whose one-shot probability of infringement is non-negligible: 0.40\% or 8.52\% in the previous example. The message is similar to our Theorem~\ref{thm:CP-counter-example}. Namely, that the NAF constructions of \cite{VyasKB23} admit models for which prompts $x$ derived from non-copyrightable aspects of $c^*$ can produce infringing generations. Based on the top half of \cite[Table 1]{li2024va3}, the prompts trivialize NAF's protection by making $k_x \approx \log(1/\safe(\subsim(c^*) ~|~ x))$.

Though the paper does not speculate, we suspect that the mechanism for the failure is similar to our construction in Theorem~\ref{thm:CP-counter-example}. The model, repeatedly fine-tuned on $c^*$, regurgitates $c^*$ with high enough probability that it survives the reweighting of $\CP$ and $\CPk$  \cite[Fig.~8]{li2024va3}.

Now we turn to the paper's limitations.

First, the attack model has a conceptual flaw, highlighting the need for our definitions-first approach. The attacker knows $c^*$ and is actively trying to induce a similar generation \cite[Sec.~4.1]{li2024va3}. Indeed, the Amplification Attack \emph{requires} knowing $c^*$. This setup allows trivial attacks, like prompting $\mathtt{return\ c^*}$. Such an attack would not be meaningful. Any infringement would be the users' fault, not something we care to prevent. 
Still, the flaw in the attack model doesn't affect the paper's experiments. Because they use a text-to-image model, he prompts used in the attack are just a few words long and contain nothing copyrightable. 

Second, the results say nothing about $k$-NAF, where a fixed $k$ upper-bounds $k_x \le k$ for every prompt $x$.  At best, the experiments are negative results for the particular $k_x$-NAF algorithm $\CPk$ given in \cite{VyasKB23}, where $k_x$ depends on $x$. But as we explain next, it is not entirely clear. (Our Theorems~\ref{thm:CP-counter-example:simplified} and~\ref{thm:NAF-counter-example} are for $k_x$-NAF and $k$-NAF algorithms, respectively.)

Third, the algorithm in \cite{li2024va3} deviates from $\CPk$ in an important way. As originally defined, $\CPk$ is a rejection-sampling version of $\CP$ (Algorithm~\ref{alg:cp}). Given a model $p$, a safe model $\safe$, and constant $k$, prompt $x$, and generation $y$, let $\rho(y|x)=\log\left(p(y|x)/\safe(y|x)\right)$.
Oversimplifying, $\CPk$ repeatedly samples $y\gets p(\cdot |x)$ until $\rho(y|x)\le k$, and returns the final sample. \cite{VyasKB23} proves that $\CPk$ is $k_x$-NAF for $k_x = k+\log(1/\nu(x))$, where $\nu(x)=\Pr_y[\rho(y|x)\le k]$.

\cite{li2024va3}'s implementation differs. Instead of fixing $k$, they fix $\nu(x)$. The experiments sample many generations $y\gets p(\cdot |x)$, and return the $\nu(x)=5\%, 10\%, \ldots$ with smallest $\rho(y|x)$. This strikes us as a meaningful difference. At a minimum, the $\CPk$ theorem doesn't apply as is. Despite the gap, we view the results of \cite{li2024va3} as strong evidence of weaknesses of $\CPk$. We conjecture that there is a value of $k_x$ for which the implemented version achieves $k_x$-NAF with high probability, but did not attempt to prove it.

\subsection{On NAF's legal relevance}
\label{app:sec:naf:legal}
NAF is motivated by two concepts from copyright law: \emph{access} and \emph{substantial similarity}.
The plaintiff in a copyright infringement claim has the burden of proving that the defendant copied original expression from the copyrighted work.
    The plaintiff does so by proving that (i) the defendant had access to the copyrighted work, and (ii) the defendant's work is substantially similar to the plaintiffs work.

With the above in mind, \cite{VyasKB23} explain the relevance of NAF to copyright liability.
\begin{quote}
    To show a copyright violation has occurred the plaintiff must prove that ``there are substantial similarities between the defendant's work and original elements of the plaintiff's work'' (assuming access). Its negation would be to show that defendant's work is not substantially similar to the original elements of the plaintiff's work. Our approach would instead correspond to showing that the defendant's work is close to a work which was produced without access to the plaintiff's work. [W]e think this is a stronger guarantee...
\end{quote}
As for why it's a stronger guarantee, the argument is as follows.
The probability that the defendant's work---produced by the real model $p$---is substantially similar to plaintiff's work is not much greater than the probability would have been had the defendant used the safe model (which had no access). We heuristically expect the latter probability to be miniscule. Hence, substantial similarity between the defendant's and plaintiff's works is exceedingly unlikely.

Elkin-Koren et al.\ \cite{elkinkoren2024copyright} correctly argue that copyright cannot be ``reduced to privacy.'' However, this is a straw man of version of \cite{VyasKB23} and the present paper (see Appendix~\ref{sec:niva} for additional discussion). 

The sharpest criticism is by Lee, Cooper, and Grimmelman who argue that NAF is simply wrong on the law \cite{lee2024talkin}. ``[NAF] is explicitly inspired by copyright's concept of access, but copyright law itself does not work that way. Just as two authors can independently create identical works and each hold a copyright in theirs, it is not a defense to copyright infringement that you would have copied the work from somewhere else if you hadn't copied it from the plaintiff.''

We agree that NAF's envisioned legal defense doesn't work (technical guarantees aside). To see why, consider a case in which the model's generated output was in fact substantially similar to a piece of training data.
As to the access element, the defendant did in fact have access to the copied work by way of the model. The defendant's work may even be so ``strikingly similar'' to the plaintiff's that access becomes moot.\footnote{%
    ``The plaintiff can prove that the defendant copied from the work by proving by a preponderance of the evidence that ... there is a striking similarity between the defendant's work and the plaintiff's copyrighted work.'' From the Ninth Circuit's \emph{Manual of Model Civil Jury Instructions} \url{https://www.ce9.uscourts.gov/jury-instructions/node/326}.}
Despite being central to NAF, \emph{access} appears to be a red herring. Whatever copyright protection NAF offers is by way of minimizing the likelihood of producing substantially similar outputs.

\section{Differential privacy background~\cite{dwork2006calibrating}}
\label{app:DP:prelims}
An algorithm is differentially private (DP) if its output never depends too much on any one unit of input data. How much is ``too much'' is governed by a parameter $\eps >0$. Smaller values of $\eps$ provide stronger guarantees. In this work, a ``unit of input data'' is one work $w$ in the training dataset $\D$. 

\begin{definition}[Neighboring datasets]
    Datasets $\D,\D' \in \W^*$ are \emph{neighboring} if they differ by inserting or deleting a single element. We denote neighboring datasets by $\D \sim \D'$. 
    \end{definition}

\begin{definition}[$(\eps,\delta)$-Differential privacy]
    Let $\eps,\delta \ge 0$, and let $M:\W^* \to \Omega$ be an algorithm mapping dataset $\D \in \W^*$ to some output domain $\Omega$. $M$ is \emph{$(\eps,\delta)$-differentially private ($(\eps,\delta)$-DP)} if for all neighboring pairs $\D, \D'\in \W^*$, and all subsets of outputs $S\sub \Omega$: 
    \[\Pr[M(\D) \in S] \le e^\eps \cdot \Pr[M(\D') \in S] + \delta.\]
\end{definition}

Anything that one does with the output of a DP algorithm is also DP. In DP parlance, DP is robust to post-processing in the presence of arbitrary auxiliary information.
Formally, for any function $f$, any string $\aux$, and any set $S$:
\begin{equation}
    \label{eq:dp-with-postprocessing}
    \Pr[f(M(\D), \aux) \in S] \le e^\eps \cdot \Pr[f(M(\D',\aux)) \in S] + \delta.
\end{equation}

\section{Proof of Theorem~\ref{thm:tainted-unclean}}
\label{app:thm:tainted-unclean}

We restate the theorem for convenience.
\begin{theorem}%
    \label{thm:app:tainted-unclean}
    Fix $\beta<1$, $n\ge 1$, and $m\ge\frac{n}{\beta}$. 
    Suppose there exists $W = \left(w_{i,j}\right) \in \W^{m\times n}$ such that for all $i \neq i' \in [m]$ and all $j \in [n]$:
    \[\ideas(w_{i,j}) = \ideas(w_{i',j}) \quad\text{and}\quad \subsim(w_{i,j}) \cap \subsim(w_{i',j}) = \emptyset.\]
    If $\train$ be tainted with respect to $\ideas$, then $\train$ is not $(\kappa,\beta)$-clean.
\end{theorem}
The proof uses Lemma~\ref{lemma:tainted:bad-d}, stated below.
\begin{proof}[Proof of Theorem~\ref{thm:app:tainted-unclean}]
    \newcommand{\h}[1]{\tilde{#1}}
    Fix $\ideas$ and let $\train$ be tainted with respect to $\ideas$. Let $\u$ be the user guaranteed by taintedness of $\train$, and $\tau$ its output distribution.
    We must show $\train$ is not $(\kappa,\beta)$-clean.
    Namely, that there exist $\h{u}$, $\h{\C}$, and $\h{D}$ for which (i) $\h{u}$ is $\beta$-blameless, and (ii) there exists $\h\aux$ such that $\h{\tau}(\subsim(\h{\C});~\h{\aux})\ge \kappa$. Here, $\h{\tau}$ is $\h{u}$'s output distribution. 

    Let $\id = \ideas(w_{1,1})$. Define $\h\u = \u_{\id}$ to be the user with $\id$ hard-coded that simulates $\u(\id)$, always ignoring its auxiliary input $\aux$. 
    By construction, $\h\tau(E;\aux) = \tau(E;\id)$ for all $\aux$ and all events $E$. That is, the event $E$ is independent of $\aux$.
    Because $\h\u$ ignores its auxiliary input, we can invoke Lemma~\ref{lemma:tainted:bad-d}. Thus, there exists $\h\D$ such that $\h\u$ is $\beta$-blameless with respect to $\h\D$, $\h\C = \h\D$, and $\train$ (condition (i) above). 
    Moreover, $w_{i,1}\in \h\D$ for some $i\in [n]$.

    Let $\id' =\ideas(w_{i,1})$. By hypothesis, $\id = \id$ and therefore $\h\u = \u_{\id'}$.
    Taintedness of $\train$ implies condition (ii) above:
    \[\h\tau(\subsim(\C);\bot) = \tau(\subsim(\C);\id') = \tau(\subsim(\D);\id')>0.99 \ge \kappa. \qedhere\]
\end{proof}

\begin{lemma}
    \label{lemma:tainted:bad-d}
    Fix $\beta<1$, $n\ge 1$, and $m\ge\frac{n}{\beta}$. 
    Suppose there exists $W = \left(w_{i,j}\right) \in \W^{m\times n}$ such that for all $i \neq i' \in [m]$ and all $j \in [n]$:
    $\subsim(w_{i,j}) \cap \subsim(w_{i',j}) = \emptyset.$    
    Let $\u$ be a user that ignores its auxiliary input: $\u^p(\aux) = \u^p(\bot)$ for all $\aux$.
    For all $\train$, there exists $x\in [m]^n$ such that $\u$ is $\beta$-blameless with respect to $\D^{(x)}=\left(w_{x_j,j}\right)_{j\in[n]}$, $\C=\D^{(x)}$, and $\train$.
\end{lemma}

\begin{proof}[Proof of Lemma~\ref{lemma:tainted:bad-d}]
    \newcommand{\Dx}{\D^{(x)}}
    \newcommand{\taux}{\tau^{(x)}}
    \newcommand{\jx}{j^*(x)}
    Because $\u$ ignores $\aux$, we omit it throughout.
    Take $\C = \D$, which implies $\C_{-\D} = \emptyset$. Hence, user $\u$ is $\beta$-blameless if for all $w \in \D$:
        $\tau_{-w}(\subsim(w))\le\beta.$

    Let $\Dx_{-j} = \scrub(\Dx,w_{x_j,j}):=(w_{x_{j'},j'})_{j'\neq j}$. (For a counter example, we can choose the copyright dependency graph and hence $\scrub$.)
    Define \[p_j(x) = \taux_{-w_{x_j,j}}(\subsim(w_{x_j,j})) = \Pr_{\substack{p \gets \train(\Dx_{-j})\\z \gets \u^p(\bot)}}[z\in \subsim(w_{x_j,j})].\]
    We must show that:
    \begin{equation}
        \exists x \in [m]^n, \forall j\in [n]: \quad p_j(x)\le \beta. \tag{$\star$}\label{eq:lemma:counting-condition}
    \end{equation}

    For $x\in [m]^n$, define $\jx = \argmax_{j\in [n]} p_j(x)$. 
    Define the sets $Z_j = \{x\in [m]^n : \jx = j\}.$
    One of the $Z_j$ contains at least $m^n/n$ distinct strings.
    Moreover, there is a subset $X^* \sub Z_j$ of at least $(m^n/n)/m^{n-1} = m/n$ strings that agree on all but the $j$th coordinate.
    
    Summarizing, the set $X^*$ satisfies three properties. (1)~For all $x,x'\in X^*$, $j^*(x) =j^*(x') =: j^*$. (2)~All $x,x'\in X^*$ disagree at the $j^*(x)$th coordinate, and agree on all other coordinates. (3)~$|X^*| \ge m/n$.

    By (2), $\D^{(x)}_{-j^*} = \D^{(x')}_{-j^*}$ for all $x,x' \in X^*$.
    It follows that $\tau^{(x)}_{-w_{x_{j^*},j^*}} = \tau^{(x')}_{-w_{j^*},x'_{j^*}} =: \tau^*$ for all $x,x' \in X^*$. 
    Consider the events $E_x = \subsim(w_{x_{j^*},j^*})$ for $x\in X^*$. By hypothesis and by (2), these events are disjoint. Hence, $\sum_{x\in X^*} \tau^*(E_x) \le 1$. By (3), there exists $x\in X^*$ such that $p_{j^*}(x) = \tau^*(E_x) \le n/m\le \beta$. By (1) and definition of $\jx$, we have for all $j\in [n]$, $p_j(x) \le p_{j^*}(x) \le \beta$. This proves~\eqref{eq:lemma:counting-condition} and completes the proof.
\end{proof}

\end{document}